\documentclass[aps,pre,preprint,floatfix]{revtex4}

\usepackage{latexsym}
\usepackage[usenames]{color}
\usepackage{amsthm}
\usepackage{hyperref}

\usepackage{amsmath}    % need for subequations
\usepackage{amsfonts}  %note how statements can be commented out
\usepackage{amssymb}
\usepackage{graphicx}
\usepackage{slashed}
\usepackage{fouridx}

\theoremstyle{plain} \newtheorem{defi}{Definition}[section]
\theoremstyle{plain} \newtheorem{prop}{Proposition}[section]
\theoremstyle{plain} \newtheorem{theor}{Theorem}[section] 
\theoremstyle{plain} \newtheorem{lemma}{Lemma}[section]
\theoremstyle{plain}  
\theoremstyle{remark} \newtheorem*{rem}{Remark}
\theoremstyle{plain} 
\theoremstyle{remark}

\newcommand\CROSS[1]{%
  \hbox{%
    \vbox{
      \hrule
      \kern2.5pt
      \hbox{$#1$\,\strut}
    }%
  \vrule
  }\mskip\thickmuskip
}

\begin{document}
\small{}
\begin{center}
\Large{\textbf{Digital calculus and finite groups \\ in quantum mechanics}}\\ 
~\\

\large{Vladimir Garc\'{\i}a-Morales}\\

\normalsize{}
~\\
Departament de Termodin\`amica, Universitat de Val\`encia, \\ E-46100 Burjassot, Spain
\\ garmovla@uv.es
\end{center}
\small{}
\noindent By means of a digit function that has been introduced in a recent formulation of classical and quantum mechanics, we provide a new construction of some infinite families of finite groups, both abelian and nonabelian, of importance for theoretical, atomic and molecular physics. Our construction is not based on algebraic relationships satisfied by generators, but in establishing the appropriate law of composition that induces the group structure on a finite set of nonnegative integers (the cardinal of the set being equal to the order of the group) thus making computations with finite groups quite straightforward. We establish the abstract laws of composition for infinite families of finite groups including all cyclic groups (and any direct sums of them), dihedral, dicyclic and other metacyclic groups, the symmetric groups of permutations of $p$ symbols and the alternating groups of even permutations. Specific examples are given to illustrate the expressions for the law of composition obtained in each case. 

\pagebreak
%We also give a simple proof of the exponential convergence to equilibrium of a simple map between finite chains of digits, in which the digits of one chain are gradually replaced by the digits of the other.

%
%
%%\small{Institute for Advanced Study -  Technische Universit\"{a}t M\"{u}nchen,\\ Lichtenbergstr. 2a, D-85748 Garching, Germany \\
%%vmorales@ph.tum.de}
%\end{center}
%\tableofcontents

%\newenvironment{sciabstract}{\begin{quote}
% \end{quote}
%}
~ \\

\normalsize{}

\section{Introduction}

Symmetry plays a fundamental role in quantum mechanics and its study is the goal of group theory \cite{Weyl, WignerBOOK,Tinkham}. Groups occurring in physics are usually infinite: The standard model of particle physics is described by the unitary product of the infinite groups $SU(3)\times SU(2)\times U(1)$. Finite, discrete groups are also important \cite{Fairbairn, Kornyak, Ishimori, Kornyak1, Smirnov, Harrison1, Joshipura} since they occur as subgroups of the infinite ones.  For example, the alternating groups $A_{4}$ and $A_{5}$, the symmetric group $\text{Sym}_{4}$ and the dihedral groups $D_{2q}$ arise as subgroups of $SU(3)$  \cite{Fairbairn, Ludl, Grimus} and have been studied, for example, in connection with the invariance of lepton masses and mixing matrices \cite{Ludl}.  

Operations on finite groups are also more naturally connected to computation, where one usually replaces the real line by the unit circle $S^{1}$ or the real Fourier transform by the fast Fourier transform \cite{Terras}. In Chemistry, they capture the structural symmetries of molecules. For example, buckminsterfullerene $C_{60}$ \cite{Baggott} is described by the icosahedral symmetry group $A_{5}\times \mathbb{Z}_{2}$, which is the direct product of the two finite groups $A_{5}$  the alternating group of 5 symbols and $\mathbb{Z}_{2}$ the cyclic group of order 2: A knowledge of $A_{5}$ (of order 60) is necessary to understand the spectral lines of $C_{60}$ \cite{Terras}. 

All finite groups are subgroups of the symmetric group of $p$ symbols $\text{Sym}_{p}$ and, hence, can be represented in terms of permutations. They are usually constructed in terms of algebraic relationships obeyed by a few group generators \cite{Coxeter2}. In this article, we present a new construction, which provides explicit expressions for the laws of composition that induce the corresponding group structure in the finite set $S$ of the nonnegative integers in the interval $[0, p-1]$, with $p >1 \in \mathbb{N}$ being the order of the group. Thus, for infinite families of finite groups, we provide an explicit mathematical expression that allows their Cayley tables to be obtained in a straightforward manner. This is accomplished for all abelian finite groups (cyclic groups and direct sums of them) as well as some important families of nonabelian ones as the metacyclic groups (including the dihedral and dicyclic groups), the symmetric groups of permutations of $p$ elements (of order $p!$) and the alternating groups of even permutations of $p$ symbols (of order $p!/2$).  

Our constructive approach makes use of a digit function which is the central concept of a formulation of classical and quantum mechanics that we have recently introduced, based on the principle of least radix economy \cite{QUANTUM}. The digit function has been further shown to naturally describe fractals \cite{arxiv3} and cellular automata \cite{comphys, VGM1, VGM2, VGM3}, leading to a closed mathematical expression for the dynamics of these systems \cite{arxiv2}. Because finite groups have found an unexpected privileged place in this mathematical approach to quantum physics \cite{arxiv3}, let us first summarize our previous results to better motivate the present effort.

The principle of least radix economy \cite{QUANTUM} has been introduced as a generalization/reinterpretation of the principle of least action, so that quantum mechanics (in its orthodox Copenhagen interpretation) is naturally encompassed together with classical physics. The unconventional idea behind this principle comes from the connection made between physics and computation: it is assumed that physical laws and physical numbers are represented by a finite (or countably infinite) alphabet. We must, however, emphasize that \emph{this does not mean finiteness of spacetime or of any physical quantity involved}: Our approach is indifferent to the question whether nature is ultimately finite or not. What we mean by a finite alphabet of symbols can be made clear if we think in our representation of numbers in the decimal base: a finite alphabet containing 10 symbols is used to represent any number. Thus, the radix fixes the alphabet size, being synonymous with it. We claim in \cite{QUANTUM} that nature works always in the appropriate, most economical radix, in each situation (independently of our choice to represent the physical numbers involved) and that this choice has dynamical consequences. In fact, we ascribe the existence of classical and quantum physics to that very fact. \emph{The alphabet size has a physical meaning} \cite{QUANTUM}. Such physical alphabet size $\eta$ (a natural number) is specified by the floor (lower closest integer) function of the dimensionless Lagrangian action, i.e $\eta=\lfloor S/h \rfloor$, where $h$ is Planck's constant. Thus, in the classical limit, the alphabet size tends to infinity and in the quantum regime only a finite number of symbols is relevant. Physical laws are then derived in a unified manner by demanding that the following quantity, called radix economy
\begin{equation}
\mathcal{C}(\eta, S/h)=\eta \lfloor 1+\log_{\eta}(S/h)\rfloor
\end{equation}
attains a minimum \cite{QUANTUM}. When $S/h$ is infinitely large, this principle reduces to the principle of least action. Thus, the physical radix $\eta$ is the appropriate radix to represent any physical real-valued physical quantity $x$, which has then a radix-$\eta$ expansion given by \cite{QUANTUM, arxiv3}
\begin{equation}
x=\text{sign}(x)\sum_{k=-\infty}^{\lfloor \log_{\eta}|x|\rfloor} \eta^{k} \mathbf{d}_{\eta}(k,|x|) \label{idenintro}
\end{equation}
Note that this is an expansion for the \emph{numerical value} of $x$. (In \cite{QUANTUM} $x$ was always considered nonnegative.) This value explicitly depends on powers of the radix $\eta$ and on the digits $\mathbf{d}_{\eta}(k,|x|)$ that accompany those powers and which are nonnegative integer numbers in the interval $[0,\eta-1]$. (For the definition of the digit function see Eq. (\ref{cucuA}) below.) \emph{Freeing physics of the convention of an artificial choice of a fixed radix brings new degrees of freedom which, constrained by the principle of least radix economy, lead to a natural physical foundation for both classical and quantum mechanics} \cite{QUANTUM}.

The digit function entering in Eq. (\ref{idenintro}) has many interesting properties. For example, if $f$ and $g$ are functions over the finite set $S$, we have $f(g(k))=\mathbf{d}_{\eta}\left(g(k),\sum_{n=0}^{\eta-1}\eta^{n}f(n)\right)$ (see Theorem II.2 below, as well as \cite{arxiv2}). This property, called the composition-decomposition theorem allows a concise general expression for all deterministic cellular automata maps  to be found in terms of the digit function \cite{arxiv2} as
\begin{equation}
x_{t+1}^{j}=\mathbf{d}_{\eta}\left(\sum_{k=-r}^{l}\eta^{k+r}x_{t}^{j+k} , R \right)
\end{equation}
where $x_{t}^{j} \in [0,\eta-1]$ is a quantity that is updated at time $t$ on a site with label $j$. $l$ and $r$ denote the number of neighbors to the left and to the right, respectively and $R$ is the Wolfram code \cite{Wolfram} of the cellular automaton rule \cite{arxiv2}. Deterministic cellular automata play a crucial role in the formulation of quantum mechanics by Gerard t'Hooft \cite{hooft, hooftNEW}.

We have also recently shown \cite{arxiv3} that \emph{any} quantity $x$ (an scalar, vector, matrix, etc.) can be further partitioned/quantized (regardless of the variables on which it depends) in terms of a linear superposition of $\lambda$ different fractal objects $x_{n}$ so that $x=\sum_{n=0}^{\lambda-1}x_{n}$. Each $x_{n}$ is given by
\begin{eqnarray}
&&x_{n}= \frac{2\ \text{sign}x}{\lambda(\lambda-1)}\sum_{k=-\infty}^{\lfloor \log_{p}|x| \rfloor}\eta^{k}\mathbf{d}_{\eta}(k,|x|)g_{\lambda}(|k|,n) \label{lampm}
\end{eqnarray}
where $g_{\lambda}(|k|,n): \mathbb{Z} \times \mathbb{Z} \to S$ is a function that maps the integers $k$ and $n$ to the finite set $S$ of the non-negative integers $\in [0,\lambda-1]$ and whose action can be specified by a table with $|k|$ given in the rows and $n$ in the columns, so that $g_{\lambda}(|k|+\lambda,n)=g_{\lambda}(|k|,n+\lambda)=g_{\lambda}(|k|,n)$ is $\lambda$-periodic on both $k$ and $n$ and such that it has the latin square property when restricted to the $p$ elements within one $\lambda$-period of both variables. \emph{The Cayley tables of finite groups, loops and semigroups of order $\lambda$ are all valid instances of $g_{\lambda}(|k|,n)$}. Thus, the fractal objects are connected by symmetry relationships and one further has $x=\sum_{n=0}^{\lambda-1}x_{n}$ for every value of the variables on which $x$ may depend. Thus \emph{any conserved quantity} (energy, charge, additive quantum numbers, etc.) \emph{can be partitioned in nontrivial ways through the mathematical methods opened by the digit function}. Since such partitioning is connected to symmetry relationships induced by the function $g_{\lambda}(|k|,n)$ in Eq. (\ref{lampm}) we undertake in this article the task of establishing the form of this function for most prominent families of finite groups. This coincides with the problem of finding the function of two integer arguments $m \in S$, $n \in S$ that gives the Cayley table of the finite group in question. We \emph{establish} such expression for infinite families of abelian and nonabelian groups. 

In the next section we discuss in more detail the digit function and some of its number theoretic and algebraic properties. Then we proceed to the construction of cyclic groups (that were briefly considered in \cite{arxiv3}), both additive and multiplicative, as well as all their direct sums, metacyclic groups (including the dihedral groups and the dycyclic groups, which themselves include e.g. the quaternion group), symmetric groups and alternating groups.

\section{The digit function and its basic properties}

The basic mathematical properties of the digit function have been already discussed in some very recent works \cite{QUANTUM, comphys, arxiv2, arxiv3}. Since we shall need most of them, we provide them here with a concise proof for the sake of completeness (the reader is also referred to \cite{arxiv2, arxiv3}).

We shall consider in this article $x$ a non-negative integer. To see how the expression that defines the digit function is obtained, we expand the numerical value of $x$ in radix $p$ as
\begin{equation}
x=a_{N}p^{N}+a_{N-1}p^{N-1}+\ldots+a_{1}p+a_{0} \label{bexpanin}
\end{equation}
where the $a_{k}$'s are all nonnegative integers $\in [0, p-1]$. This representation is unique for each nonnegative integer \cite{Andrews}. Now, note that 
\begin{equation}
\left \lfloor \frac{x}{p^{k}} \right \rfloor=a_{N}p^{N-k}+a_{N-1}p^{N-i-1}+\ldots +a_{k+1}p+a_{k} 
\end{equation}
and
\begin{equation}
p\left \lfloor \frac{x}{p^{k+1}} \right \rfloor=a_{N}p^{N-k}+a_{N-1}p^{N-k-1}+\ldots +a_{k+1}p 
\end{equation}
whence, by subtracting both equations, we obtain the digit $a_{k}$. 

\begin{defi} The digit function, for $p \in \mathbb{N}$ and $k, x$ nonnegative integers is defined as
\begin{equation}
\mathbf{d}_{p}(k,x)\equiv \left \lfloor \frac{x}{p^{k}} \right \rfloor-p\left \lfloor \frac{x}{p^{k+1}} \right \rfloor  \qquad (\equiv a_{k} \text{ above})  \label{cucuA}
\end{equation}
and gives the $k$-th digit of $x$ written in radix $p$. If $p=1$ the digit function returns zero.
\end{defi}

%This function is of fundamental importance in digital calculus, from which it takes its name. It will be the subject of extensive analysis in Chapter \ref{digitalcal}.

With the digit function we can express Eq. (\ref{bexpanin}) as 
\begin{equation}
x=\sum_{k=0}^{N} p^{k} \mathbf{d}_{p}(k,x) \label{idenin}
\end{equation}

\begin{prop} \cite{QUANTUM} \label{totalno} One has $N=\lfloor \log_{p}x \rfloor$ in Eq. (\ref{bexpanin}).  
\end{prop}

\begin{proof} 
This can be easily seen from the fact that since $p > 1$ and $
\lfloor \log_{p}x \rfloor \le  \log_{p}x  < \lfloor \log_{p}x \rfloor+1$ we have $p^{\lfloor \log_{p}x \rfloor} \le  x  < p^{\lfloor \log_{p}x \rfloor+1}$ and therefore $\left \lfloor \frac{x}{p^{k}} \right \rfloor=0 \qquad \forall k > \lfloor \log_{p}x \rfloor$. From the definition Eq. (\ref{cucuA}) this in turn implies that the sum is bounded from above by 
$\lfloor \log_{p}x \rfloor$. Thus,
\begin{equation}
 \mathbf{d}_p(k,x) = 0 \quad \qquad \forall k > \lfloor \log_{p}x \rfloor  \label{bound}  
\end{equation}
which indicates that the total number of digits is equal to $\lfloor \log_{p}x \rfloor+1=N+1$.
\end{proof}

\begin{lemma} The digit function satisfies (for $n$ and $m$ nonnegative integers)
\begin{eqnarray}
\mathbf{d}_{p}(k,x+np^{k+1}) &=& \mathbf{d}_{p}(k,x) \label{pro1}\\
\mathbf{d}_{p}(k,p^{k}x) &=&\mathbf{d}_{p}(0,x) \label{pro3} \\
\mathbf{d}_{p}(0,\mathbf{d}_{p}(k,x)) &=& \mathbf{d}_{p}(k,x) \label{pro2} \\
\mathbf{d}_{p}(0,n+\mathbf{d}_{p}(0,m)) &=& \mathbf{d}_{p}(0,n+m) \label{pro2b} \\
\mathbf{d}_{p}(0,n\mathbf{d}_{p}(0,m)) &=& \mathbf{d}_{p}(0,nm) \label{pro2c} \\
 \mathbf{d}_p(n,p^{m}) &=& \mathbf{d}_p(m,p^{n})= \delta_{mn}  \label{krone} \\
 \mathbf{d}_p(k,n) &=& \sum_{j=0}^{m>n}\mathbf{d}_p(k,j)\mathbf{d}_p(j,p^{n}) \label{expan}
\end{eqnarray}
where $\delta_{mn}$ is Kronecker's delta (it returns $1$ if $m=n$ and zero otherwise).
\end{lemma}

%The following relationship also holds
%\begin{eqnarray}
%&& \mathbf{d}_p(i,A) = \sum_{j=1}^{B}\mathbf{d}_p(i,j)\mathbf{d}_p(j,p^{A-1}) \quad (\forall B \ge A)  \label{cuttie}
%\end{eqnarray}
%and, furthermore
%\begin{equation}
 
\begin{proof} Eq. (\ref{pro1}) and (\ref{pro3}) follow directly from the definition \cite{arxiv3}. 

Eq. (\ref{pro2}) follows by noting that $\mathbf{d}_{p}(k,x)$ is an integer $0 \le \mathbf{d}_{p}(k,x) \le p-1$. Then 
\begin{equation}
\mathbf{d}_{p}(0,\mathbf{d}_{p}(k,x))=\mathbf{d}_{p}(k,x)-p\left \lfloor \frac{\mathbf{d}_{p}(k,x)}{p} \right \rfloor=\mathbf{d}_{p}(k,x)
\end{equation}
since $\left \lfloor \mathbf{d}_{p}(k,x)/p \right \rfloor=0$. Eq. (\ref{pro2b}) is proved by using Euclidean division, since we can always write $m=ap+b$ with $a, b$ integers and $b=\mathbf{d}_{p}(0,m)$. Therefore
\begin{equation}
\mathbf{d}_{p}(0,n+\mathbf{d}_{p}(0,m))=\mathbf{d}_{p}(0,n+ap+\mathbf{d}_{p}(0,m))=\mathbf{d}_{p}(0,n+m)  
\end{equation}
where Eq. (\ref{pro1}) has also been used. Eq. (\ref{pro2c}) is also proved in a similar way, using the distributive property of ordinary addition and multiplication as well. Eq. (\ref{krone}) is simply proved from the definition and Eq. (\ref{expan}) is a simple consequence of it.
\end{proof}

Eq. (\ref{pro2c}) is closely related to the following classical result, found in almost every textbook on elementary number theory.

\begin{theor} \textbf{\emph{(Euclid's Lemma.)}} \label{eulem} If a prime number $p$ divides the product $ab$ of two nonnegative numbers $a$ and $b$ then $p$ divides $a$ or $b$ (or both). 
\end{theor}

\begin{proof} If $p$ divides $ab$ we have
\begin{equation}
\mathbf{d}_{p}(0,ab)=0
\end{equation}
from which, by using Eq. (\ref{pro2c}) twice
\begin{equation}
\mathbf{d}_{p}(0,\mathbf{d}_{p}(0,a)\mathbf{d}_{p}(0,b))=0
\end{equation}
Now since both numbers $\mathbf{d}_{p}(0,a)$ and $\mathbf{d}_{p}(0,b)$ are both nonnegative integers lower than $p$, their product cannot be equal to $p$, because $p$ is prime. Therefore, one necessarily has either $\mathbf{d}_{p}(0,a)=0$ or $\mathbf{d}_{p}(0,b)=0$ (or both).
\end{proof}

The following simple result allows to split into parts any composition of functions over finite sets. This is a most important property of the digit function, discovered in \cite{arxiv2}. It brings the dynamics of cellular automata \cite{Wolfram, VGM1, VGM2, VGM3} to a closed form, and it shall be used below in deriving the symmetric group $\text{Sym}_{p}$. 

\begin{theor} \textbf{\emph{(Composition-decomposition theorem. \cite{arxiv3})}} \label{compy} Let  $g:S \to S$ and $f: S \to S$ be two functions on the finite set of the integers in the interval $[0,p-1]$ ($p>1 \in \mathbb{N}$) and $n \in S$. Then
\begin{equation}
f(g(n))=\mathbf{d}_{p}\left(g(n), \sum_{k=0}^{p-1}p^{k}f(k)\right)=\left \lfloor \frac{\sum_{k=0}^{p-1}p^{k}f(k)}{p^{g(n)}} \right \rfloor-p \left \lfloor \frac{\sum_{k=0}^{p-1}p^{k}f(k)}{p^{g(n)+1}} \right \rfloor \label{condecon}
%=\left \lfloor \frac{R}{p^{g(n)}} \right \rfloor-p \left \lfloor \frac{R}{p^{g(n)+1}} \right \rfloor
\end{equation}
\end{theor}

\begin{proof} From Eq. (\ref{pro3}) we have
\begin{equation}
\mathbf{d}_{p}\left(g(n), \sum_{k=0}^{p-1}p^{k}f(k)\right)=\mathbf{d}_{p}\left(0, \sum_{k=0}^{p-1}p^{k-g(n)}f(k)\right)
\end{equation}
This function extracts the zeroth digit of the quantity $\sum_{k=0}^{p-1}p^{k-g(n)}f(k)$, i.e. the value of $f(k)$ for which the exponent of the accompanying power of $p$ within the sum is zero. But this happens when $k=g(n)$ from which the result follows.
\end{proof}

\begin{lemma} Let $n\in \mathbb{N}$. Further properties of the digit function are
\begin{eqnarray}
\mathbf{d}_{np}(k,x)&=&\mathbf{d}_{p}\left(k, \frac{x}{n^{k}} \right)+p\mathbf{d}_{n}\left(k, \frac{x}{p^{k+1}} \right)=\mathbf{d}_{n}\left(k, \frac{x}{p^{k}} \right)+n\mathbf{d}_{p}\left(k, \frac{x}{n^{k+1}} \right)    \label{reldi1} \\
0&=& \prod_{m=0}^{p-1}\left[\mathbf{d}_{p}(k,x)-m\right]    \label{reldi2} \\
\mathbf{d}_{p}(k,x)&=&\left \lfloor p \left \{ \frac{x}{p^{k+1}} \right \} \right \rfloor=  p \left \{ \frac{x}{p^{k+1}} \right \}-\left \{ \frac{x}{p^{k}} \right \}   \label{cucufrac} 
\end{eqnarray}
\begin{eqnarray}
k&=&\mathbf{d}_{p}\left(k,\ p^{p}-\frac{p^{p}-p}{(p-1)^{2}} \right) \label{Galoprop} \\
\sum_{k=0}^{\lfloor \log_{p}n \rfloor} \mathbf{d}_p(k,n)&=&
n-(p-1)\sum_{k=1}^{\lfloor \log_{p}n \rfloor} \left \lfloor \frac{n}{p^{k}} \right \rfloor \qquad \qquad \qquad \qquad \qquad \qquad \qquad \qquad \qquad \label{flory}
\end{eqnarray}
where $\{ \ldots \}$ denotes the fractional part. 
\end{lemma}

\begin{proof} By using the definition, we have
\begin{eqnarray}
\mathbf{d}_{p}\left(k, \frac{x}{n^{k}} \right)+p\mathbf{d}_{n}\left(k, \frac{x}{p^{k+1}} \right)&=&
\left \lfloor \frac{x}{p^{k}n^{k}} \right \rfloor -p \left \lfloor \frac{x}{p^{k+1}n^{k}} \right \rfloor
+p \left \lfloor \frac{x}{p^{k+1}n^{k}} \right \rfloor-np\left \lfloor \frac{x}{p^{k+1}n^{k+1}} \right \rfloor \nonumber \\
&=& \left \lfloor \frac{x}{(np)^{k}} \right \rfloor -np\left \lfloor \frac{x}{(np)^{k+1}} \right \rfloor =\mathbf{d}_{np}(k,x)
\end{eqnarray}
which proves Eq. (\ref{reldi1}). To prove Eq. (\ref{reldi2}) is also trivial by realizing the fact that 
$\mathbf{d}_{p}(k,x)$ takes only any of the integer values between 0 and $p-1$ and all these values are scanned by the product. 

Eq. (\ref{cucufrac}) follows also from straightforward calculations, by using that $y=\lfloor y \rfloor+\{ y \}$ for any real number $y$. Thus
\begin{eqnarray}
\mathbf{d}_{p}(k,x)&=& \left \lfloor \frac{x}{p^{k}} \right \rfloor -p \left \lfloor \frac{x}{p^{k+1}} \right \rfloor = \left \lfloor p\frac{x}{p^{k+1}} \right \rfloor -\left \lfloor p \left \lfloor \frac{x}{p^{k+1}} \right \rfloor \right \rfloor = \left \lfloor p \left \{ \frac{x}{p^{k+1}} \right \} \right \rfloor \nonumber \\ 
&=& \frac{x}{p^{k}}  - \left \{ \frac{x}{p^{k}} \right \} -p\frac{x}{p^{k+1}}+p \left \{ \frac{x}{p^{k+1}} \right \}=p \left \{ \frac{x}{p^{k+1}} \right \}-\left \{ \frac{x}{p^{k}} \right \}     \nonumber
\end{eqnarray}
To prove Eq. (\ref{Galoprop}) note that
\begin{equation}
\sum_{k=0}^{p-1}kp^{k}= p^{p}-\frac{p^{p}-p}{(p-1)^{2}} 
\end{equation}
This, hence, means that the $k$-th digit of $p^{p}-\frac{p^{p}-p}{(p-1)^{2}}$ is equal to $k$ in radix $p$, which is nothing but Eq. (\ref{Galoprop}). Finally Eq. (\ref{flory}) can also be easily proved because the sum telescopes
\begin{eqnarray}
&&\sum_{k=0}^{\lfloor \log_{p}n \rfloor} \mathbf{d}_p(k,n) = \sum_{k=0}^{\lfloor \log_{p}n \rfloor} \left(\left\lfloor \frac{n}{p^{k}} \right \rfloor -p\left \lfloor \frac{n}{p^{k+1}} \right \rfloor \right)
\nonumber \\
&&= 
\left \lfloor \frac{n}{p^{0}} \right \rfloor-p\left \lfloor \frac{n}{p^{\lfloor \log_{p}n \rfloor+1}} \right \rfloor-(p-1)\sum_{k=1}^{\lfloor \log_{p}n \rfloor} \left \lfloor \frac{n}{p^{k}} \right \rfloor=n-(p-1)\sum_{k=1}^{\lfloor \log_{p}n \rfloor} \left \lfloor \frac{n}{p^{k}} \right \rfloor \nonumber
\end{eqnarray} 
When $p$ is prime the l.h.s. of Eq. (\ref{flory}) corresponds to the exponent with which $p$ appears in the prime factorization of $A!$ \cite{Andrews}.
\end{proof}

We shall use the following result in constructing direct sums of finite groups.
%\subsection{Mixed-radix numeral systems and the fundamental theorem of arithmetic}

%This proves Eq. (\ref{reldi1}), if one also takes into account that $n$ and $p$ can be exchanged for the l.h.s. is symmetric on $n$ and $p$. Eqs. (\ref{reldi2}), (\ref{reldi3a}) and (\ref{reldi3b}) then follow trivially, by noticing that $0 \le \mathbf{d}_{\min (n,p)}\left(i, \frac{x}{(\max (n,p))^{j}} \right) \le \max (n,p)-1$.  
\begin{theor} \label{numeralsys}
Let $p_{0}$, $p_{1}$,... $p_{N-1} \in \mathbb{N}$. The following relationship holds 
\begin{equation}
\mathbf{d}_{p_{0}p_{1}\ldots p_{N-1}}\left(k, x \right)=\sum_{h=0}^{N-1}\mathbf{d}_{p_{h}}\left(k, \frac{x}{(\prod_{m=0}^{h-1}p_{m})^{k+1}(\prod_{n=h+1}^{N-1}p_{n})^{k}}\right)\prod_{j=0}^{h-1}p_{j} \label{numeralgen}
\end{equation}
\end{theor}

\begin{proof} For $N=1$ the result is trivially valid and for $N=2$ it reduces to Eq. (\ref{reldi1}). Let us assume the result valid for $N$ factors. Then, for $N+1$ factors we have, by using Eq. (\ref{reldi1})
\begin{eqnarray}
&&\mathbf{d}_{p_{0}p_{1}\ldots p_{N}}\left(k, x \right)=\mathbf{d}_{p_{0}p_{1}\ldots p_{N-1}}\left(k, \frac{x}{p_{N}^{k}} \right)+p_{0}p_{1}\ldots p_{N-1}\mathbf{d}_{p_{N}}\left(k, \frac{x}{(p_{0}p_{1}\ldots p_{N-1})^{k+1}} \right) \nonumber \\
&&=\sum_{h=0}^{N-1}\mathbf{d}_{p_{h}}\left(k, \frac{x/p_{N}^{k}}{(\prod_{m=0}^{h-1}p_{m})^{k+1}(\prod_{n=h+1}^{N-1}p_{n})^{k}}\right)\prod_{k=0}^{h-1}p_{k}+\mathbf{d}_{p_{N}}\left(k, \frac{x}{(\prod_{m=0}^{N-1}p_{k})^{k+1}} \right)\prod_{j=0}^{N-1}p_{j} \nonumber \\
&&=\sum_{h=0}^{N}\mathbf{d}_{p_{h}}\left(k, \frac{x}{(\prod_{m=0}^{h-1}p_{m})^{k+1}(\prod_{n=h+1}^{N}p_{n})^{k}}\right)\prod_{j=0}^{h-1}p_{j} \qedhere \nonumber
\end{eqnarray}
\end{proof}
Theorem \ref{numeralsys} systematically allows mixed-radix numeral systems to be constructed. 
In order to see how, let us consider, for simplicity, $x=n$ be an integer number of $N$ digits, such that $n \in [0,P-1]$ where $P=p_{0}\cdot p_{1}\cdot p_{N-1}$. Then $p_{h}$ can be chosen as radix for the $h$-th digit of $n$. Therefore, from the theorem above we have, since $n < P$
\begin{equation}
n=\mathbf{d}_{P}\left(0, n \right)=
\mathbf{d}_{p_{0}p_{1}\ldots p_{N-1}}\left(0, n \right)=\sum_{h=0}^{N-1}\mathbf{d}_{p_{h}}\left(0, \frac{n}{\prod_{m=0}^{h-1}p_{m}}\right)\prod_{k=0}^{h-1}p_{k}
\end{equation}
and, thus, the $h$-th digit of $n$ in this mixed-radix numeral system is given by
\begin{equation}
\mathbf{d}_{p_{h}}\left(0, \frac{n}{\prod_{m=0}^{h-1}p_{m}}\right)
\end{equation}
An example of mixed-radix numeral system is provided by the factorial number system \cite{KnuthII, Knuth}, also called factoradic. Such system provides a Gray code that can be put in a one-to-one correspondence with the permutations of $N$ distinct elements \cite{Knuth}. In this system, the less significant digit is always 0, and $p_{0}$ is simply taken as $1$ and we have $p_{h}=h+1$ for $h \in [0,N-1]$. Therefore, in this numeral system the $h$-th digit is given by
\begin{equation}
\mathbf{d}_{h+1}\left(0, \frac{n}{h!}\right)
\end{equation}

The fundamental theorem of arithmetic \cite{Hardy}  is usually proved by invoking Euclid's lemma  above. We give here another simple proof based in Theorem \ref{numeralsys}. Although the existence part is not original and invokes induction as usual, the uniqueness part constitutes a successful application of Theorem \ref{numeralsys}.

\begin{theor} \textbf{\emph{(Fundamental theorem of arithmetic.)}} \label{fundari} Every integer $n$ greater than 1 either is prime itself or is the product of prime numbers. Although the order of the primes in the second case is arbitrary, the primes themselves are not. 
\end{theor}

\begin{proof} We need to show that the decomposition exists, i.e. that every integer $n >1$ is a product of primes, and that this product is unique. The first part is proved by induction. We first assume it is true for all numbers between $1$ and $n$. If $n$ is prime, there is nothing more to prove and we have just a 'product' of one prime factor. If $n$ is not a prime, there are integers $a$ and $b$, such that $n = ab$ and $1 < a \le b < n$. By the induction hypothesis, $a = p_{0}p_{1}\ldots p_{N-1}$ and $b = q_{0}q_{1}\ldots q_{M-1}$ are products of primes. But then $n = ab = p_{0}p_{1}\ldots p_{N-1}q_{0}q_{1}\ldots q_{M-1}$ is also a product of primes.

We now show that the prime number decomposition is unique. Let us assume that $n$ has decomposition $p_{0}p_{1}\ldots p_{N-1}$ and that there exists another possible decomposition $p_{0}'p_{1}'\ldots p_{N-1}'$ as well. From Eq. (\ref{numeralgen}), by taking $P=n=p_{0}p_{1}\ldots p_{N-1}$, we have  
 \begin{equation}
\mathbf{d}_{p_{0}p_{1}\ldots p_{N-1}}\left(0, n \right)=\sum_{h=0}^{N-1}\mathbf{d}_{p_{h}}\left(0, \frac{n}{\prod_{m=0}^{h-1}p_{m}}\right)\prod_{j=0}^{h-1}p_{j}=0 \label{numeralproof}
\end{equation}
since $\mathbf{d}_{p_{0}p_{1}\ldots p_{N-1}}\left(k, n \right)=\mathbf{d}_{n}\left(k, n \right)=0$. This then necessarily means that all prefactors in the terms within the sum of the right hand side must be zero separately (otherwise the corresponding term in the sum would contribute a positive quantity), i.e.
\begin{equation}
\mathbf{d}_{p_{h}}\left(0, \frac{n}{\prod_{m=0}^{h-1}p_{m}}\right)=0 \label{mandator1}
\end{equation}
Let us now replace in this expression the other representation $n=p_{0}'p_{1}'\ldots p_{N-1}'$. We observe, however, that in order for Eq. (\ref{mandator1}) to hold, $n$ must contain all factors $p_{h}$, so that $\frac{n}{\prod_{m=0}^{h-1}p_{m}}$ is an integer multiple of $p_{h}$. Then necessarily, there is a $p_{j}'=p_{h}$ for every $h$. Which, in turn, means that the decomposition $n=p_{0}p_{1}\ldots p_{N-1}$ is unique (up to reordering of the factors). 
\end{proof}

\section{Construction of the abelian finite groups} \label{finagrou}

\subsection{The cyclic group $C_{p}$}

We now focus in the finite set $S$ of nonnegative integers $0,\ 1,\ \ldots,\ p-1$ with $p>1 \in \mathbb{N}$ and construct operations such that this finite set is endowed with the group structure. 

%To fix the notation, we first introduce a definition that we shall need.
%
%\begin{defi} Let $G$ be a group under the operation $+$ and let $S$ be a group under a certain operation $*$. An application $f:G \to S$ such that for $x, y \in G$ one has
%\begin{equation}
%f(x+y)=f(x)*f(y)
%\end{equation}
%is called a \textbf{\emph{group homomorphism}} from $G$ under $+$ to $S$ under $*$. 
%\end{defi}

The following theorem makes explicit an important consequence of the fact that the first integer digit $\mathbf{d}_{p}(0,x)$ governs the divisibility of a nonnegative integer $x$, giving the remainder upon division modulo $p$. Thus, such function implements the rules of modular arithmetic 
%and the following result, proved also in \cite{arxiv3} is not surprising (but it is necessary for all what follows and it is given here for the sake of completeness).

%\subsection{Abelian groups and finite fields}

\begin{theor} \cite{arxiv3} \label{cyc} Let $S$ be the finite set of nonnegative integers $0,\ 1,\ \ldots,\ p-1$ with $p>1 \in \mathbb{N}$ and let $m, n \in S$. Under the operation
\begin{equation}
m+_{p}n \equiv \mathbf{d}_{p}\left(0, m+n \right)=g_{p}(m,n)  \label{adimod}
\end{equation}  
the integers in $S$ constitute the finite cyclic group $C_{p}$. Furthermore, the function $f: x\in \mathbb{Z} \to S$, $f(x)=\mathbf{d}_{p}(0,x)$ is also a group homomorphism between all integers $\mathbb{Z}$ under ordinary addition and the group $C_{p}$. 
\end{theor}

\begin{proof} The operation Eq. (\ref{adimod}) is nothing but the addition modulo $p$ of $n$ and $m$ which is well known to have the cyclic group structure and proving it amounts to a simple exercise in elementary algebra. Although the proof can be found in \cite{arxiv3}, because of its importance for what follows we give it here explicitly also for completeness.
We show that all group axioms are satisfied, together with the properties that the group is abelian and that is generated by a single element, of order $p$. 
\begin{itemize}
\item \emph{1. Closure}: The digit function yields the remainder of $m+n$ under division by $p$ which is, of course, an integer $\in [0, p-1]$ as well.
\item \emph{2. Associative property}: We have 
\begin{equation}
\mathbf{d}_{p}\left(0, m+\mathbf{d}_{p}\left(0, n+k \right)\right)=\mathbf{d}_{p}\left(0, m+n+k \right)=\mathbf{d}_{p}\left(0, \mathbf{d}_{p}\left(0, m+n \right)+k \right) \nonumber
\end{equation}
where Eq. (\ref{pro2b}) has been used twice.
\item \emph{3. Neutral element}: The neutral element is 0, as in the ordinary sum.
\item \emph{4. Inverse element}: The inverse element of $m$ is $p-m$, since $\mathbf{d}_{p}\left(0, m+p-m \right)=0$.
\end{itemize}
Thus, the group axioms are satisfied. Furthermore
\begin{itemize}
\item $\mathbf{d}_{p}\left(0, m+n \right)=\mathbf{d}_{p}\left(0, n+m \right)$, the group is abelian.
\item A single element generates the whole group. Let us just consider $n=1$ and operate repeatedly with it. For a group of order $p$ we have, from Eq. (\ref{pro1}) 
\begin{equation}
\mathbf{d}_{p}\left(0, 1+1 \right)=2, \ \ldots, \ \mathbf{d}_{p}\left(0, 1+p-2 \right)=p-1, \ \mathbf{d}_{p}\left(0, 1+p-1 \right)=0,\  \mathbf{d}_{p}\left(0, 1+p \right)=1  \nonumber
\end{equation} 
The orbit of the element $1$ has order $p$ equal to the one of the group and the group is cyclic. 
\end{itemize}
Let now $x$ and $y$ be any integers. We have, from Eq. (\ref{pro2b})
\begin{equation}
\mathbf{d}_{p}\left(0, x+y \right)=\mathbf{d}_{p}\left(0, \mathbf{d}_{p}\left(0, x \right)+\mathbf{d}_{p}\left(0, y \right) \right)=\mathbf{d}_{p}\left(0, x \right)+_{p}\mathbf{d}_{p}\left(0, y \right)
\end{equation} 
and this expression explicitly shows that the function constitutes the group homomorphism  stated in the theorem.
\end{proof}

The following remarkable facts are a consequence of the latin square property of the Cayley table.

\begin{prop} The following identities hold $\forall m \in \mathbb{Z}$, 
\begin{eqnarray}
\sum_{n=0}^{p-1}\mathbf{d}_{p}\left(0, m+n \right)&=&\frac{p(p-1)}{2} \label{simplead} \\
\prod_{\forall n \in S,\ \mathbf{d}_{p}\left(0, m+n \right) \ne 0}\mathbf{d}_{p}\left(0, m+n \right)&=&(p-1)! \label{simpleprod}
\end{eqnarray}
\end{prop}

\begin{proof} Since the integers in $[0, p-1]$ (the set $S$) constitute a cyclic group under $\mathbf{d}_{p}\left(0, m+n \right)$, by fixing $m\in [0, p-1]$ the sum in Eq. (\ref{simplead}) runs over all elements in $S$: Each column and row of the Cayley table of a group  contains each of its distinct elements only once (latin square property). Therefore
\begin{equation}
\sum_{n=0}^{p-1}\mathbf{d}_{p}\left(0, m+n \right)=\sum_{k=0}^{p-1}k=\frac{p(p-1)}{2} \nonumber 
\end{equation}
Eq. (\ref{simpleprod}) is also trivial because of the same reason 
\begin{equation}
\prod_{\forall n \in S,\ \mathbf{d}_{p}\left(0, m+n \right) \ne 0}\mathbf{d}_{p}\left(0, m+n \right)=\prod_{k=1}^{p-1}k=(p-1)! \nonumber \qedhere 
\end{equation}
\end{proof}

%\begin{prop} \label{splitad} Let $m$ and $n$ be integers $\in S$. The following relationship holds
%\begin{equation}
%\mathbf{d}_{p}\left(0, m+n\right)=\mathbf{d}_{p}\left(0, m\right)+\mathbf{d}_{p}\left(0, n\right)-p\mathbf{d}_{p}\left(1, m+n \right) \label{deconmod}
%\end{equation}
%\end{prop}
%
%\begin{proof} Either $0 \le m+n \le p-1$ or $p \le m+n \le 2(p-1)$. This depends on whether $\mathbf{d}_{p}\left(1, m+n \right)$ is zero or one respectively. In the second case one needs to subtract $p$ in order to have $m+n \mod p$, i.e. $\mathbf{d}_{p}\left(0, m+n\right)$, which is always lower or equal than $p-1$. This is what the term $-p\mathbf{d}_{p}\left(1, m+n \right)$ in Eq. (\ref{deconmod}) does.
%\end{proof}

The Cayley tables of $C_{2}$, $C_{3}$ and $C_{4}$ are provided below as an example. The values of $m$ increase from top to bottom and the values of $n$ from left to right, in both cases from $0$ to $p-1$. Thus, the table provides $\mathbf{d}_{p}\left(0, m+n \right)$ for $p=2, 3$ and $4$ respectively.  The latter can also be seen as a function on the integer lattice $\mathbb{Z}\times \mathbb{Z}$ by virtue of the existing homomorphism (we shall discuss this point in detail later)

\begin{center}
\begin{tabular}{c|cc}
%\hline 
$\ \ C_{2}  \ \ $ & $ \ 0  \ $ & $\ \  1 \ \ $ \\
\hline
$\ 0 \ $ & $\ 0 \ $ & $\  1 \ $ \\
$\ 1 \ $ & $\ 1 \ $ & $ \ 0 \ $ \\
%\hline
\end{tabular}
\qquad 
\begin{tabular}{c|ccc}
%\hline 
$\ C_{3}  \ $ & $ 0 \ $ & $\ 1 \ $ & $\   2 \  $ \\
\hline
$\ 0 $ & $\ 0 \ $ & $\ 1 \ $ & $\ 2 \ $ \\
$\ 1 $ & $\ 1 \ $ & $\ 2 \ $ & $\  0 \ $ \\
$\ 2 $ & $\ 2 \ $ & $\ 0 \ $ & $ \ 1 \ $ \\
%\hline
\end{tabular}
\qquad 
\begin{tabular}{c|cccc}
%\hline 
$\ C_{4}  \ $ & $ 0 \ $ & $\ 1 \ $ & $\   2 \ $ & $ \ 3 \ $\\
\hline
$\ 0 $ & $\ 0 \ $ & $\ 1 \ $ & $\ 2 \ $ & $\ 3 \ $\\
$\ 1 $ & $\ 1 \ $ & $\ 2 \ $ & $\  3 \ $ & $\ 0 \ $\\
$\ 2 $ & $\ 2 \ $ & $\ 3 \ $ & $ \ 0 \ $ & $\ 1 \ $\\
$\ 3 $ & $\ 3 \ $ & $\ 0 \ $ & $ \ 1 \ $ & $\ 2 \ $\\
%\hline
\end{tabular}
\end{center}

Note that, because of Eqs. (\ref{reldi1}), by also using Eq. (\ref{pro3}) we have
\begin{eqnarray}
\mathbf{d}_{4}(0,m+n)&=&\mathbf{d}_{2}\left(0, m+n \right)+2\mathbf{d}_{2}\left(0, \frac{m+n}{2} \right) \nonumber \\  
&=&\mathbf{d}_{2}\left(0, m+n \right)+2\mathbf{d}_{2}\left(1, m+n \right)
\label{C4}
\end{eqnarray}
as can be easily checked in the tables. This is of course different to the direct sum $C_{2}\oplus C_{2}$ which is provided by the operation
\begin{eqnarray}
\mathbf{d}_{2}\left(0, m+n \right)+2\mathbf{d}_{2}\left(0, m+n \right)
\label{C2oC2}
\end{eqnarray}
with Cayley table
\begin{center}
\begin{tabular}{c|cccc}
%\hline 
$\ C_{2}\oplus C_{2}  \ $ & $ 0 \ $ & $\ 1 \ $ & $\   2 \ $ & $ \ 3 \ $\\
\hline
$\ 0 $ & $\ 0 \ $ & $\ 1 \ $ & $\ 2 \ $ & $\ 3 \ $\\
$\ 1 $ & $\ 1 \ $ & $\ 0 \ $ & $\  3 \ $ & $\ 2 \ $\\
$\ 2 $ & $\ 2 \ $ & $\ 3 \ $ & $ \ 0 \ $ & $\ 1 \ $\\
$\ 3 $ & $\ 3 \ $ & $\ 2 \ $ & $ \ 1 \ $ & $\ 0 \ $\\
%\hline
\end{tabular}
\end{center}

\subsection{Direct sums of cyclic groups}

The following theorem establishes a general result to systematically construct any direct sum of cyclic groups.

\begin{theor} \label{dsumcyc} Let $P=p_{0}\cdot p_{1}\cdot \ldots \cdot p_{N-1}$ with $p_{h\in \mathbb{Z}} \ge 1 $ ($h \in [0,N-1]$), $p_{h} \in \mathbb{N}$. Let $m, n \in \mathbb{Z}$ be in the interval $[0,P-1]$. The digits $h$ of $m$ and $n$ under the operation
\begin{equation}
\mathbf{d}_{p_{h}}\left(0, \mathbf{d}_{p_{h}}\left(0, \frac{m}{\prod_{j=0}^{h-1}p_{j}}\right)
+\mathbf{d}_{p_{h}}\left(0, \frac{n}{\prod_{j=0}^{h-1}p_{j}}\right)\right) \label{adimodsubpro}
\end{equation}  
constitute the subgroup $C_{p_{h}}$ of the direct sum $C_{p_{0}}\oplus C_{p_{1}}\oplus\ldots \oplus C_{p_{N-1}}$ formed by all integers in the interval $[0,P-1]$ under the operation
\begin{equation}
g_{P}(m,n)\equiv  \sum_{h=0}^{N-1}\mathbf{d}_{p_{h}}\left(0, \mathbf{d}_{p_{h}}\left(0, \frac{m}{\prod_{j=0}^{h-1}p_{j}}\right)
+\mathbf{d}_{p_{h}}\left(0, \frac{n}{\prod_{j=0}^{h-1}p_{j}}\right)\right)\prod_{j=0}^{h-1}p_{j}
 \label{adimodpro}
\end{equation} 
\end{theor}

\begin{figure*} %Shown is a window where $t \in [0,150]$ and $i\in [1, 150]$
\includegraphics[width=0.5 \textwidth]{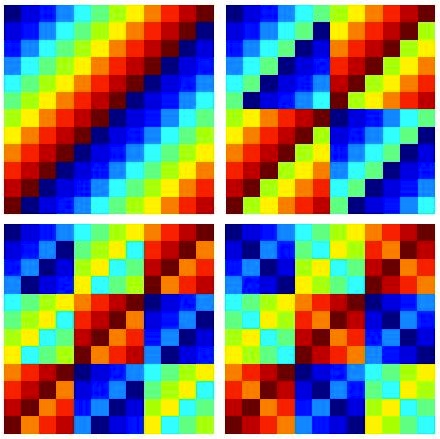}
\caption{\scriptsize{Cayley tables of the finite cyclic groups $C_{12}$ (top left), $C_{2}\oplus C_{6}$ (top right), $C_{3}\oplus C_{4}$ (bottom left), $C_{2}\oplus C_{2}\oplus C_{3}$ (bottom right), all generated by Eq. (\ref{adimodpro}). The first row in each panel ranges from 0 to 11 (from left to right) establishing the integer numerical values of the color code.}} \label{dsums}
\end{figure*}

\begin{proof} The proof of this theorem constitutes an application of Theorem \ref{numeralsys}. Since both $m$ and $n$ are integers in the interval $[0,P-1]$ with $P=p_{0}\cdot p_{1}\cdot p_{N-1}$ we can choose each $p_{h}$ for the $h$-th digit of these integers and represent them in a mixed-radix numeral system. Thus, we have
\begin{eqnarray}
n&=&\mathbf{d}_{P}\left(0, n \right)=
\mathbf{d}_{p_{0}p_{1}\ldots p_{N-1}}\left(0, n \right)=\sum_{h=0}^{N-1}\mathbf{d}_{p_{h}}\left(0, \frac{n}{\prod_{j=0}^{h-1}p_{j}}\right)\prod_{j=0}^{h-1}p_{j} \nonumber \\
m&=&\mathbf{d}_{P}\left(0, n \right)=
\mathbf{d}_{p_{0}p_{1}\ldots p_{N-1}}\left(0, n \right)=\sum_{h=0}^{N-1}\mathbf{d}_{p_{h}}\left(0, \frac{m}{\prod_{j=0}^{h-1}p_{j}}\right)\prod_{j=0}^{h-1}p_{j} \nonumber
\end{eqnarray}
from which we realize that Eq. (\ref{adimodsubpro}) is the addition modulo $p_{h}$ of the $h$-th digits of $m$ and $n$ when written in this mixed-radix system. Such digits are independent of any other digits of $m$ and $n$ and by Theorem \ref{cyc} such operation has the structure of a cyclic group.

Thus we, also realize that Eq. (\ref{adimodpro}) constitutes the bitwise addition of $m$ and $n$ modulo the radix at the position of each digit in the mixed-radix numeral system, where each digit is an independent cyclic subgroup under such operation. Eq. (\ref{adimodpro}) yields a number in $[0,P-1]$ (closure) and the set of such numbers is easily proved to have the group structure of the direct sum in the statement of the theorem.
\end{proof}

In Fig. (\ref{dsums}), Eq. (\ref{adimodpro}) is plotted in the plane on a $12 \times 12$ square for 
the finite cyclic groups $C_{12}$ (top left), $C_{2}\oplus C_{6}$ (top right), $C_{3}\oplus C_{4}$ (bottom left), $C_{2}\oplus C_{2}\oplus C_{3}$ (bottom right). These plots correspond to the  Cayley tables of the above mentioned groups. The first row in each panel ranges from 0 to 11 (from left to right) establishing the integer numerical values of the color code. 

The so-called Fundamental Theorem of Finite Cyclic Groups \cite{Gallian} establishes that any finite abelian group is a direct sum of finite cyclic groups (up to isomorphism). Such theorem has in this work the following important implication, which we state without further proof. 

\begin{theor} \label{fundamentalcy} Eq. (\ref{adimodpro}) gives the Cayley table of any finite abelian group up to permutation of rows and columns, once the  coefficients $p_{0}, \ldots, \ p_{N-1}$ $\in \mathbb{N}$ are specified.
\end{theor}

We also note now that $m$ and $n$ in Eq. (\ref{adimodpro}) can indeed be replaced by any integers in the plane. Then, the whole plane is filled with the Cayley table produced in the domain $0 \le m \le p-1$, $0 \le n \le p-1$ repeated as a motif in both directions an infinite number of times, as shown in Fig. \ref{homo}.  This simply follows from the construction, by using Eq. (\ref{pro1}) in Eq. (\ref{adimodpro}).

%This simply follows because, from the construction, by using Eq. (\ref{pro1})
%\begin{equation}
%\mathcal{C}(n+k_{n}P^{j_{n}},m)=\mathcal{C}(n,m+k_{m}P^{j_{m}})=\mathcal{C}(n+k_{n}P^{j_{n}},m+k_{m}P^{j_{m}})=\mathcal{C}(n,m)
%\end{equation}
%with $k_{n}$, $k_{m}$, $j_{n}$, $j_{m}$ being any integers.

\begin{figure*} %Shown is a window where $t \in [0,150]$ and $i\in [1, 150]$
\includegraphics[width=0.7 \textwidth]{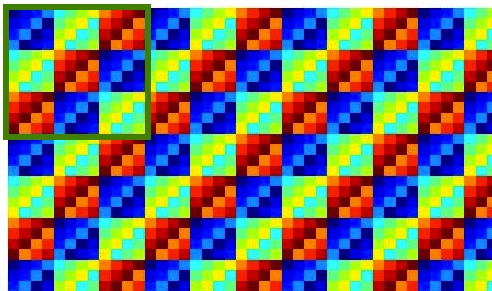}
\caption{\scriptsize{By replacing $m$ and $n$ in Eq. (\ref{adimodpro}) by any integers an infinite lattice can be constructed so that any Cayley table is used as motif in both directions. Here, an example is given by using $C_{3}\oplus C_{4}$ as base motif (in the green box).}} \label{homo}
\end{figure*}

\subsection{The multiplicative group and the finite Galois field $\mathbb{F}_{p}$}

We have considered above the groups under addition modulo $p$. We now consider the multiplicative groups. We remove the element $0$ out of the group and consider now multiplication modulo $p$.

\begin{figure*} %Shown is a window where $t \in [0,150]$ and $i\in [1, 150]$
\includegraphics[width=0.5 \textwidth]{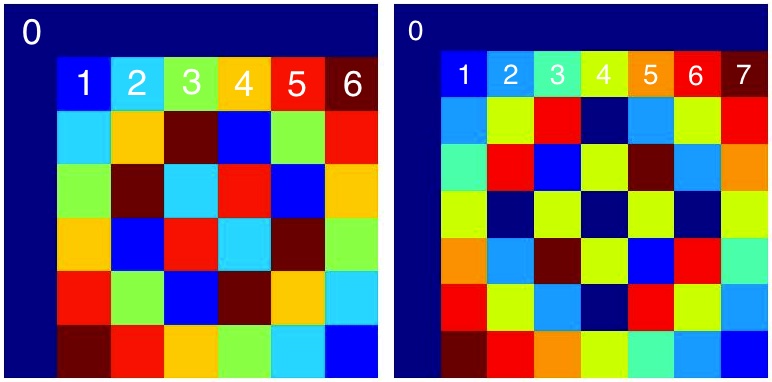}
\caption{\scriptsize{Cayley tables of $\left(\mathbb{Z}/7\mathbb{Z}\right)^{\times}$ (left) and $\left(\mathbb{Z}/8\mathbb{Z}\right)^{\times}$ (right), obtained from Eq. (\ref{prodmod}) for $p=7$ and $p=8$ on squares of $7 \times 7$ and $8 \times 8$ respectively. The numerical values of the color code are indicated in the figure.}} \label{mult}
\end{figure*}

\begin{theor} \label{prog} Let $m, n \in [0,p-1]$ be integers. Under the operation
\begin{equation}
\mathbf{d}_{p}\left(0, mn \right) \label{prodmod}
\end{equation}  
the subset of $S$, $S^{*}$, where $0$ is excluded, constitute a finite abelian group (called the multiplicative group $\left(\mathbb{Z}/p\mathbb{Z}\right)^{\times}$) if and only if $p$ is prime.
\end{theor}

\begin{proof} We prove first that $S^{*}$ being a group under Eq. (\ref{prodmod}) implies $p$ prime. Let us assume that $p$ is non prime. Then we can find two integers $m$, $n$ $\in [1, p-1]$ such that $mn=p$. This means $\mathbf{d}_{p}\left(0, mn \right)=\mathbf{d}_{p}\left(0, p \right)=0$, which clearly fails to provide inverses for certain elements. Thus the premise is false and $p$ must be prime.

Now we prove that $p$ prime implies that $S^{*}$ is a finite abelian group under the operation in Eq. (\ref{prodmod}). The closure property is trivial since $\mathbf{d}_{p}\left(0, mn \right)$ yields an integer number $\in [0,p-1]$ and zero is ruled out by the fact that it is not possible to have $mn=p$ ($p$ is prime) and neither $m$ nor $n$ can be zero. The associative property follows trivially from the one of the ordinary multiplication and Eq. (\ref{pro2c}). The neutral element is $1$. Finally, each element has a unique inverse: note that $\gcd (m,p)=1$ for any $m \in S$ since $p$ is prime. This means that, from B\'ezout's identity, one can always find integers $x$ and $y$ such that
\begin{equation} 
\gcd (m,p)=mx+py=1
\end{equation}
Thus, we have
\begin{equation} 
1=\mathbf{d}_{p}(0,1)=\mathbf{d}_{p}(0,\gcd (m,p))=\mathbf{d}_{p}(0,mx+py)=\mathbf{d}_{p}(0,mx) \label{multinvers}
\end{equation}
which means that for each $m \in S$ there exists $x=\mathbf{d}_{p}(0,x) \in S$ such that $\mathbf{d}_{p}(0,mx)=1$. Such $x$ is the unique inverse element of $m$ and hence $x=m^{-1}$. The operation is also trivially abelian. This completes the proof.
\end{proof}

For prime $p$, the additive group $C_{p}$ and the multiplicative one $\left(\mathbb{Z}/p\mathbb{Z}\right)^{\times}$ define the Galois field $\mathbb{F}_{p}$. The number $p$ is called the characteristic of the field. 

%\begin{lemma} Let $S=[0,p-1]$ and $m \in S'=[1,p-1]$, $m \in \left(\mathbb{Z}/p\mathbb{Z}\right)^{\times}$ with $p$ prime. Then
%\begin{equation}
%\mathbf{d}_{p}(0,m^{p})=m
%\end{equation} 
%\end{lemma}
%\begin{proof} 
%Since $p$ is prime, $\left(\mathbb{Z}/p\mathbb{Z}\right)^{\times}$ is a multiplicative finite group of order $p-1$ and neutral element $'1'$, and one necessarily has $\mathbf{d}_{p}(0,m^{p-1})=\mathbf{d}_{p}(0,1)=1$ since $m$ has the same order as the one of the group. Thus, the result easily follows since: $\mathbf{d}_{p}(0,m^{p})=\mathbf{d}_{p}(0,m^{p-1}m)=\mathbf{d}_{p}(0,1\cdot m)=m$. This is an alternative proof of Fermat's 'Little' Theorem Eq. (\ref{Fermalit}).
%\end{proof}

In Fig. \ref{mult} the Cayley tables of $\left(\mathbb{Z}/7\mathbb{Z}\right)^{\times}$ (left) and $\left(\mathbb{Z}/8\mathbb{Z}\right)^{\times}$ (right) are shown, obtained from Eq. (\ref{prodmod}) for $p=7$ and $p=8$. The tables are thus square of sizes $7 \times 7$ and $8 \times 8$ respectively. The numerical values of the color code are indicated in the figure. Clearly, while $\left(\mathbb{Z}/7\mathbb{Z}\right)^{\times}$ has a group structure if one excludes $0$, $\left(\mathbb{Z}/8\mathbb{Z}\right)^{\times}$ has not: For the pairs of integers $(2,4)$, $(4,2)$, $(4,4)$, $(4,6)$ and $(6,4)$, Eq. (\ref{prodmod}) returns zero.

\section{Construction of  nonabelian finite groups}

We have found above the law of composition for all abelian finite groups, Eq. (\ref{adimodpro}). We now study the construction of some important infinite families of nonabelian groups that occur in theoretical physics. Some other families of great interest, as the Chevalley groups, modular groups and sporadic simple groups (Mathieu groups) are under construction and shall be discussed elsewhere.

\subsection{The dihedral groups $D_{2q}$}

Let $p=2q	$ be an even number. We can construct in $S$ the direct sum $C_{q} \oplus C_{2}$ by means of Eq. (\ref{adimodpro}). This results in the operation
\begin{equation}
\mathbf{d}_{q}\left(0, m+n\right)+q\mathbf{d}_{2}\left(0, \mathbf{d}_{2}\left(0,\frac{m}{q}\right)+\mathbf{d}_{2}\left(0,\frac{n}{q}\right)\right) \label{disum2}
\end{equation}
This operation involves the two numbers $m$ and $n$ $\in S$. From the construction, such numbers are each written in the mixed radix numeral system with radices $q$ and $2$ by means of two digits, the most significant one being '0' or '1' and the less significant one being between 0 and $q-1$. Then Eq. (\ref{disum2}) produces a number $\in S$ with  each of its two digits resulting from the addition modulo 2 and modulo $q$ of the most and less significant digits of $m$ and $n$ respectively.

By noting that the inverse of an element $n \in S$ is $p-n$ also $\in S$ (because of the group structure of the direct sum) we can now construct a non-abelian group in the following straightforward way. Let $*$ denote the operation that induces in $S$ the nonabelian group structure. Now let $m*n$ represent the result of this operation acting on $m$ and $n$ both in $S$. If the most significant digit of $m$ is zero, we can take $m*n$ as given by Eq. (\ref{disum2}) above. However, if the most significant digit of $m$ is 1, we take the additive inverse $p-\mathbf{d}_{q}(0,n)$ of the least significant digit of $n$ and add it to the less significant digit of $m$ modulo $q$ (the most significant digits being added modulo 2 as in the direct sum). With this operation, we break the abelian character of the direct sum Eq. (\ref{disum2}) above while still preserving the group axioms. These considerations give an sketch of the proof of the following theorem.

\begin{theor} \label{dih} Let $S$ be the finite set of nonnegative integers $0,\ 1,\ \ldots,\ p-1$ with $p=2q$, $q  \in \mathbb{N}$ and let $m, n \in S$. Under the operation
\begin{equation}
g_{p}(m,n) \equiv \mathbf{d}_{q}\left(0, m+n\cdot(-1)^{\mathbf{d}_{2}\left(0,\frac{m}{q}\right)}\right)+q\mathbf{d}_{2}\left(0, \mathbf{d}_{2}\left(0,\frac{m}{q}\right)+\mathbf{d}_{2}\left(0,\frac{n}{q}\right)\right) \label{dihsum}
\end{equation}  
the integers in $S$ constitute a finite nonabelian group: the \emph{dihedral group} $D_{2q}$. 
\end{theor}

Let us construct with help of Eq. (\ref{dihsum}) the smallest dihedral groups. For the cases $q=1$ and $q=2$ $D_{2}$ and $D_{4}$ are trivially isomorphic to $C_{2}$ and $C_{2}\oplus C_{2}$. However for $q \ge 2$, the group is nonabelian. For $q=3$ we have

\begin{center}
\begin{tabular}{c|cccccc}
%\hline 
$\  D_{6}  \ $ & $ 0 \ $ & $\ 1 \ $ & $\  2 \  $ & $\ 3 \ $ & $\ 4 \ $ & $\   5 \  $ \\
\hline
               $\ 0  \ $ & $\ 0 \ $ & $\ 1 \ $ & $\  2 \  $ & $\ 3 \ $ & $\ 4 \ $ & $\   5 \  $ \\
               $\ 1  \ $ & $\ 1 \ $ & $\ 2 \ $ & $\   0 \  $ & $\ 4 \ $ & $\ 5 \ $ & $\   3 \  $ \\
	       $\ 2  \ $ & $\ 2 \ $ & $\ 0 \ $ & $\   1 \  $ & $\ 5 \ $ & $\ 3 \ $ & $\  4 \  $ \\
	       $\ 3  \ $ & $\ 3 \ $ & $\ 5 \ $ & $\  4 \  $ & $\ 0 \ $ & $\ 2 \ $ & $\   1 \  $ \\
	       $\ 4  \ $ & $\ 4 \ $ & $\ 3 \ $ & $\   5 \  $ & $\ 1 \ $ & $\ 0 \ $ & $\  2 \  $ \\
	       $\ 5  \ $ & $\ 5 \ $ & $\ 4 \ $ & $\  3 \  $ & $\ 2 \ $ & $\ 1 \ $ & $\  0 \  $ \\
%\hline
\end{tabular}
\end{center}

This group corresponds to the group of automorphisms (rotations and reflections) that leave invariant an equilateral triangle. For $q=4$ we have

 \begin{center}
\begin{tabular}{c|cccccccc}
%\hline 
$\  D_{8}  \ $ & $ 0 \ $ & $\ 1 \ $ & $\  2 \  $ & $\ 3 \ $ & $\ 4 \ $ & $\   5 \  $ & $\ 6 \ $ & $\   7 \  $ \\
\hline
               $\ 0  \ $ & $ 0 \ $ & $\ 1 \ $ & $\  2 \  $ & $\ 3 \ $ & $\ 4 \ $ & $\   5 \  $ & $\ 6 \ $ & $\   7 \  $ \\
               $\ 1  \ $ & $ 1 \ $ & $\ 2 \ $ & $\  3 \  $ & $\ 0 \ $ & $\ 5 \ $ & $\   6 \  $ & $\ 7 \ $ & $\   4 \  $ \\
	       $\ 2  \ $ & $ 2 \ $ & $\ 3 \ $ & $\  0 \  $ & $\ 1 \ $ & $\ 6 \ $ & $\   7 \  $ & $\ 4 \ $ & $\   5 \  $ \\
	       $\ 3  \ $ & $ 3 \ $ & $\ 0 \ $ & $\  1 \  $ & $\ 2 \ $ & $\ 7 \ $ & $\   4 \  $ & $\ 5 \ $ & $\   6 \  $ \\
	       $\ 4  \ $ & $ 4 \ $ & $\ 7 \ $ & $\  6 \  $ & $\ 5 \ $ & $\ 0 \ $ & $\   3 \  $ & $\ 2 \ $ & $\   1 \  $ \\
	       $\ 5  \ $ & $ 5 \ $ & $\ 4 \ $ & $\  7 \  $ & $\ 6 \ $ & $\ 1 \ $ & $\   0 \  $ & $\ 3 \ $ & $\   2 \  $ \\
	       $\ 6  \ $ & $ 6 \ $ & $\ 5 \ $ & $\  4 \  $ & $\ 7 \ $ & $\ 2 \ $ & $\   1 \  $ & $\ 0 \ $ & $\   3 \  $ \\
	       $\ 7  \ $ & $ 7 \ $ & $\ 6 \ $ & $\  5 \  $ & $\ 4 \ $ & $\ 3 \ $ & $\   2 \  $ & $\ 1 \ $ & $\   0 \  $ \\
%\hline
\end{tabular}
\end{center}
which is the group of automorphisms of a square. In general $D_{2q}$ is the group of rotations and reflections of a regular polygon of $q$ sides. Remarkably, Eq. (\ref{dihsum}), gives the abstract definition of the infinite family of dihedral groups, making also possible to explicitly operate with them on separate pairs of elements $m$ and $n$ $\in S$ once $q$ is given. 

\subsection{The dicyclic groups $Q_{4q}$}

Let $p=4q	$ be an even number and let us consider the dihedral group $D_{4q}$, which is, by Eq. (\ref{dihsum}) given by
\begin{equation}
\mathbf{d}_{2q}\left(0, m+n\cdot(-1)^{\mathbf{d}_{2}\left(0,\frac{m}{2q}\right)}\right)+2q\mathbf{d}_{2}\left(0, \mathbf{d}_{2}\left(0,\frac{m}{2q}\right)+\mathbf{d}_{2}\left(0,\frac{n}{2q}\right)\right) \label{dihsum2}
\end{equation}
Starting from this group we can now construct the dicylic groups $Q_{4q}$. These are provided by the operation $g_{p}(m,n)$ for $m, n \in [0,p-1]$ ($p=4q$) given by 
\begin{equation}
\mathbf{d}_{2q}\left(0, m+n\cdot(-1)^{\mathbf{d}_{2}\left(0,\frac{m}{2q}\right)}+q\mathbf{d}_{2}\left(0,\frac{m}{2q}\right)\mathbf{d}_{2}\left(0,\frac{n}{2q}\right)   \right)+2q\mathbf{d}_{2}\left(0, \mathbf{d}_{2}\left(0,\frac{m}{2q}\right)+\mathbf{d}_{2}\left(0,\frac{n}{2q}\right)\right) \label{dicy}
\end{equation} 
Proving that the group axioms hold again in this case and that the group is noncommutative is a mere (but tedious) exercise. What we have done is introducing a further twist to the corresponding expression for the dihedral group which is only relevant when both most significant digits of $m$ and $n$ in the binary radix are equal to unity, in which case the value $q$ is added. This operation can be easily shown not to affect the group axioms, but the resulting group is not isomorphic to the previously constructed oned.

Let us illustrate Eq. (\ref{dicy}) by giving explicitly the resulting outcomes for the smallest dicyclic groups. For the case $q=1$, $Q_{4}$ is trivially isomorphic to $D_{4}$ which, in turn, is isomorphic to $C_{2}\oplus C_{2}$ as well. However for $q \ge 2$, the group is nonabelian. For $q=2$ we have, from Eq. (\ref{dicy})

\begin{center}
\begin{tabular}{c|cccccccc}
%\hline 
$\  Q_{8}  \ $ & $ 0 \ $ & $\ 1 \ $ & $\  2 \  $ & $\ 3 \ $ & $\ 4 \ $ & $\   5 \  $ & $\ 6 \ $ & $\   7 \  $ \\
\hline
               $\ 0  \ $ & $ 0 \ $ & $\ 1 \ $ & $\  2 \  $ & $\ 3 \ $ & $\ 4 \ $ & $\   5 \  $ & $\ 6 \ $ & $\   7 \  $ \\
               $\ 1  \ $ & $ 1 \ $ & $\ 2 \ $ & $\  3 \  $ & $\ 0 \ $ & $\ 5 \ $ & $\   6 \  $ & $\ 7 \ $ & $\   4 \  $ \\
	       $\ 2  \ $ & $ 2 \ $ & $\ 3 \ $ & $\  0 \  $ & $\ 1 \ $ & $\ 6 \ $ & $\   7 \  $ & $\ 4 \ $ & $\   5 \  $ \\
	       $\ 3  \ $ & $ 3 \ $ & $\ 0 \ $ & $\  1 \  $ & $\ 2 \ $ & $\ 7 \ $ & $\   4 \  $ & $\ 5 \ $ & $\   6 \  $ \\
	       $\ 4  \ $ & $ 4 \ $ & $\ 7 \ $ & $\  6 \  $ & $\ 5 \ $ & $\ 2 \ $ & $\   1 \  $ & $\ 0 \ $ & $\   3 \  $ \\
	       $\ 5  \ $ & $ 5 \ $ & $\ 4 \ $ & $\  7 \  $ & $\ 6 \ $ & $\ 3 \ $ & $\   2 \  $ & $\ 1 \ $ & $\   0 \  $ \\
	       $\ 6  \ $ & $ 6 \ $ & $\ 5 \ $ & $\  4 \  $ & $\ 7 \ $ & $\ 0 \ $ & $\   3 \  $ & $\ 2 \ $ & $\   1 \  $ \\
	       $\ 7  \ $ & $ 7 \ $ & $\ 6 \ $ & $\  5 \  $ & $\ 4 \ $ & $\ 1 \ $ & $\   0 \  $ & $\ 3 \ $ & $\   2 \  $ \\
%\hline
\end{tabular}
\end{center}
This is the famous quaternion group. This group gives the multiplication table of the units of the quaternion number system. To see this let $U$ denote the set of units of the quaternion number system. If we consider the bijection $f:S \to U$ given by
\begin{equation}
f(n)=\delta_{0n}+i\delta_{1n}-\delta_{2n}-i\delta_{3n}+j\delta_{4n}+ij\delta_{5n}-j\delta_{6n}-ij\delta_{7n}
\end{equation}
by applying this function to all elements above one obtains the table
\begin{center}
\begin{tabular}{c|cccccccc}
%\hline 
$\  Q_{8}  \ $ & $ 1 \ $ & $\ i \ $ & $\  -1 \  $ & $\ -i \ $ & $\ j \ $ & $\   ij \  $ & $\ -j \ $ & $\   -ij \  $ \\
\hline
               $\ 1  \ $ & $ 1 \ $ & $\ i \ $ & $\  -1 \  $ & $\ -i \ $ & $\ j \ $ & $\   ij \  $ & $\ -j \ $ & $\   -ij \  $ \\
               $\ i  \ $ & $ i \ $ & $\ -1 \ $ & $\  -i \  $ & $\ 1 \ $ & $\ ij \ $ & $\   -j \  $ & $\ -ij \ $ & $\   j \  $ \\
	       $\ -1  \ $ & $ -1 \ $ & $\ -i \ $ & $\  1 \  $ & $\ i \ $ & $\ -j \ $ & $\   -ij \  $ & $\ j \ $ & $\   ij \  $ \\
	       $\ -i  \ $ & $ -i \ $ & $\ 1 \ $ & $\  i \  $ & $\ -1 \ $ & $\ -ij \ $ & $\   j \  $ & $\ ij \ $ & $\   -j \  $ \\
	       $\ j  \ $ & $ j \ $ & $\ -ij \ $ & $\  -j \  $ & $\ ij \ $ & $\ -1 \ $ & $\   i \  $ & $\ 1 \ $ & $\   -i \  $ \\
	       $\ ij  \ $ & $ ij \ $ & $\ j \ $ & $\  -ij \  $ & $\ -j \ $ & $\ -i \ $ & $\   -1 \  $ & $\ i \ $ & $\   1 \  $ \\
	       $\ -j  \ $ & $ -j \ $ & $\ ij \ $ & $\  j \  $ & $\ -ij \ $ & $\ 1 \ $ & $\   -i \  $ & $\ -1 \ $ & $\   i \  $ \\
	       $\ -ij  \ $ & $ -ij \ $ & $\ -j \ $ & $\  ij \  $ & $\ j \ $ & $\ i \ $ & $\   1 \  $ & $\ -i \ $ & $\   -1 \  $ \\
%\hline
\end{tabular}
\end{center}
which corresponds to the multiplication table of the quaternion units. It can be checked in the table that the following identities are satisfied 
\begin{equation}
i^{2} = j^{2} = k^{2} = ijk = -1 \label{Hamilton}
\end{equation}
(where $k\equiv ij$). These constitute the fundamental formulae for quaternion multiplication discovered by W. R. Hamilton as he walked by the Brougham Bridge in Dublin on the 16th of October 1843 \cite{Baez, ConwayOCTONIONS}. Quaternion multiplication is associative but noncommutative. It is easily seen that the units of the complex numbers ($1$, $i$, $-1$, $-i$) constitute a subgroup (indeed a normal one) of the quaternions.

For, $q=3$, we have, from Eq. (\ref{dicy}) the Cayley table for $Q_{12}$ given below
\begin{center}
\begin{tabular}{c|cccccccccccc}
%\hline 
$\  Q_{12}  \ $ & $ 0 \ $ & $\ 1 \ $ & $\  2 \  $ & $\ 3 \ $ & $\ 4 \ $ & $\   5 \  $ & $\ 6 \ $ & $\   7 \  $ & $\ 8 \ $ & $\   9 \  $ & $\ 10 \ $ & $\   11  \  $ \\
\hline
               $\ 0  \ $ & $ 0 \ $ & $\ 1 \ $ & $\  2 \  $ & $\ 3 \ $ & $\ 4 \ $ & $\   5 \  $ & $\ 6 \ $ & $\   7 \  $ & $\ 8 \ $ & $\   9 \  $ & $\ 10 \ $ & $\   11  \  $ \\
               $\ 1  \ $ & $ 1 \ $ & $\ 2 \ $ & $\  3 \  $ & $\ 4 \ $ & $\ 5 \ $ & $\   0 \  $ & $\ 7 \ $ & $\   8 \  $ & $\ 9 \ $ & $\   10 \  $ & $\ 11 \ $ & $\  6  \  $ \\
	       $\ 2  \ $ & $ 2 \ $ & $\ 3 \ $ & $\  4 \  $ & $\ 5 \ $ & $\ 0 \ $ & $\   1 \  $ & $\ 8 \ $ & $\   9 \  $ & $\ 10 \ $ & $\   11 \  $ & $\ 6 \ $ & $\   7  \  $ \\
	       $\ 3  \ $ & $ 3 \ $ & $\ 4 \ $ & $\  5 \  $ & $\ 0 \ $ & $\ 1 \ $ & $\   2 \  $ & $\ 9 \ $ & $\   10 \  $ & $\ 11 \ $ & $\   6 \  $ & $\ 7 \ $ & $\   8  \  $ \\
	       $\ 4  \ $ & $ 4 \ $ & $\ 5 \ $ & $\  0 \  $ & $\ 1 \ $ & $\ 2 \ $ & $\   3 \  $ & $\ 10 \ $ & $\ 11 \  $ & $\ 6 \ $ & $\   7 \  $ & $\ 8 \ $ & $\   9  \  $ \\
	       $\ 5  \ $ & $ 5 \ $ & $\ 0 \ $ & $\  1 \  $ & $\ 2 \ $ & $\ 3 \ $ & $\   4 \  $ & $\ 11 \ $ & $\   6 \  $ & $\ 7 \ $ & $\  8 \  $ & $\ 9 \ $ & $\   10  \  $ \\
	       $\ 6  \ $ & $ 6 \ $ & $\ 11 \ $ & $\  10 \  $ & $\ 9 \ $ & $\ 8 \ $ & $\   7 \  $ & $\ 3 \ $ & $\   2 \  $ & $\ 1 \ $ & $\   0 \  $ & $\ 5 \ $ & $\   4  \  $ \\
	       $\ 7  \ $ & $ 7 \ $ & $\ 6 \ $ & $\  11 \  $ & $\ 10 \ $ & $\ 9 \ $ & $\   8 \  $ & $\ 4 \ $ & $\   3 \  $ & $\ 2 \ $ & $\   1 \  $ & $\ 0 \ $ & $\   5  \  $ \\
	       $\ 8  \ $ & $ 8 \ $ & $\ 7 \ $ & $\  6 \  $ & $\ 11 \ $ & $\ 10 \ $ & $\   9 \  $ & $\ 5 \ $ & $\   4 \  $ & $\ 3 \ $ & $\   2 \  $ & $\ 1 \ $ & $\   0  \  $ \\
	       $\ 9  \ $ & $ 9 \ $ & $\ 8 \ $ & $\  7 \  $ & $\ 6 \ $ & $\ 11 \ $ & $\   10 \  $ & $\ 0 \ $ & $\   5 \  $ & $\ 4 \ $ & $\   3 \  $ & $\ 2 \ $ & $\   1  \  $ \\
	       $\ 10  \ $ & $ 10 \ $ & $\ 9 \ $ & $\  8 \  $ & $\ 7 \ $ & $\ 6 \ $ & $\   11 \  $ & $\ 1 \ $ & $\   0 \  $ & $\ 5 \ $ & $\   4 \  $ & $\ 3 \ $ & $\   2  \  $ \\
	       $\ 11  \ $ & $ 11 \ $ & $\ 10 \ $ & $\  9 \  $ & $\ 8 \ $ & $\ 7 \ $ & $\   6 \  $ & $\ 2 \ $ & $\   1 \  $ & $\ 0 \ $ & $\   5 \  $ & $\ 4 \ $ & $\   3  \  $ \\
%\hline
\end{tabular}
\end{center}

\begin{figure*} %Shown is a window where $t \in [0,150]$ and $i\in [1, 150]$
\includegraphics[width=0.6 \textwidth]{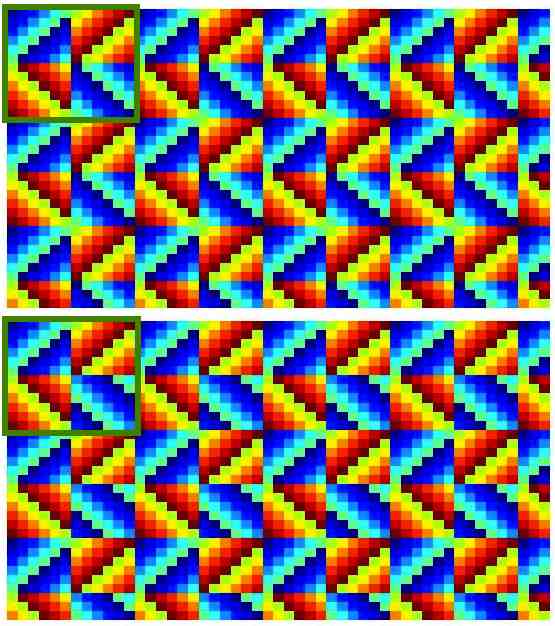}
\caption{\scriptsize{The infinite integer lattice $\mathbb{Z} \times \mathbb{Z}$ is homeomorphically mapped to the set $S$ under the operation $g_{p}(m,n)$ given by Eqs. (\ref{dihsum}) (top) and (\ref{dicy}) (bottom). By replacing any integers in the plane $m$ and $n$  the Cayley tables of $D_{12}$ (top) and $Q_{12}$ (bottom) are, respectively, obtained (in the green boxes) and repeated as motif in both horizontal and vertical directions in the plane.}} \label{homodi}
\end{figure*}

Eqs. (\ref{dihsum}) and (\ref{dicy}) not only provide the operation $g_{p}(m,n)$ such that $S$ has the structure of a dihedral or a dicyclic group respectively. They also provide respective group homomorphisms that map the integers $\mathbb{Z}$ under ordinary addition to $S$ under $g_{p}(m,n)$. 

In Fig. \ref{homodi} we observe how the integer lattice $\mathbb{Z} \times \mathbb{Z}$ is homeomorphically mapped to the Cayley table of the groups $D_{12}$ (top) and $Q_{12}$ (bottom).

 \subsection{Metacyclic groups}

Cyclic, dihedral and dicyclic groups are all themselves particular instances of a broader class of generally non-commutative finite groups called metacyclic groups. These groups are cyclic extensions of cyclic groups and have order $p=qac$. Let, $S$ be, as above, the set of integers between $0$ and $p-1$. Let $m$ and $n$ belong to that set. The set has then the structure of a metacyclic group under the operation  
\begin{equation}
\mathbf{d}_{aq}\left(0, m+n\cdot r^{\mathbf{d}_{c}\left(0,\frac{m}{qa}\right)}+q\mathbf{d}_{a}\left(0,\frac{m}{qa}\right)\mathbf{d}_{a}\left(0,\frac{n}{qa}\right)   \right)+aq\mathbf{d}_{c}\left(0, \mathbf{d}_{c}\left(0,\frac{m}{qa}\right)+\mathbf{d}_{c}\left(0,\frac{n}{qa}\right)\right) \label{metacy}
\end{equation}  
where $r$ is a number such that $\gcd (r,q)=1$. For $r=-1$ and $c=2$, Eq. (\ref{metacy}) reduces to Eq. (\ref{dihsum}) if $a=1$ and to Eq. (\ref{dicy}) if $a=2$. Cyclic groups are obtained if both $a=c=1$.

\begin{figure*} %Shown is a window where $t \in [0,150]$ and $i\in [1, 150]$
\includegraphics[width=0.65 \textwidth]{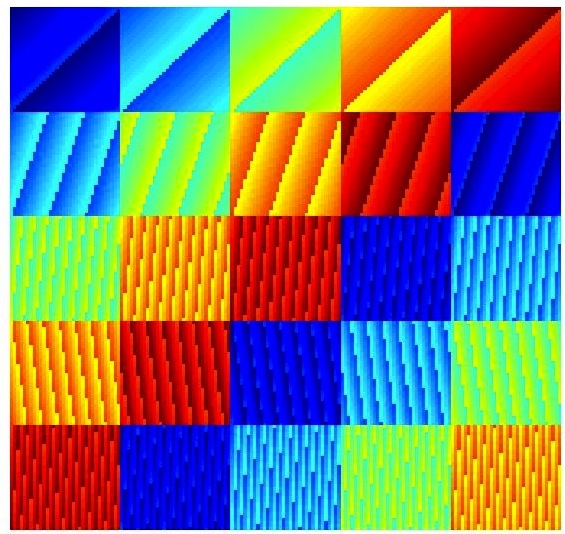}
\caption{\scriptsize{A representation of the Cayley table of the metacyclic group of order 170 with $q=17$, $c=5$, $r=3$ and $a=2$.}} \label{metcy}
\end{figure*}

In Fig. \ref{metcy} the Cayley table of the metacyclic group of order 170 with $q=17$, $c=5$, $r=3$ and $a=2$ is displayed in a color coded image with the colors ranging from 0 (dark blue) to 169 (dark red). The Cayley table of the metacyclic group $M_{16}$ with $q=8$, $c=2$, $r=5$ and $a=1$ as calculated from Eq. (\ref{metacy}) is also shown below. 

\begin{center}
\begin{tabular}{c|cccccccccccccccc}
%\hline 
$\  M_{16}  \ $ & $0\ $ & $1\ $ & $2\ $ & $3\ $ & $4\ $ & $5\ $ & $6\ $ & $7\ $ & $8\ $ & $9\ $ & $10\ $ & $11\ $ & $12\ $ & $13\ $ & $14\ $ & $15 \  $ \\
\hline
               $\ 0  \ $ & $0\ $ & $1\ $ & $2\ $ & $3\ $ & $4\ $ & $5\ $ & $6\ $ & $7\ $ & $8\ $ & $9\ $ & $10\ $ & $11\ $ & $12\ $ & $13\ $ & $14\ $ & $15 \  $ \\
               $\ 1  \ $ & $1\ $ & $2\ $ & $3\ $ & $4\ $ & $5\ $ & $6\ $ & $7\ $ & $0\ $ & $9\ $ & $10\ $ & $11\ $ & $12\ $ & $13\ $ & $14\ $ & $15\ $ & $8 \  $ \\
	       $\ 2  \ $ & $2\ $ & $3\ $ & $4\ $ & $5\ $ & $6\ $ & $7\ $ & $0\ $ & $1\ $ & $10\ $ & $11\ $ & $12\ $ & $13\ $ & $14\ $ & $15\ $ & $8\ $ & $9 \  $ \\
	       $\ 3  \ $ & $3\ $ & $4\ $ & $5\ $ & $6\ $ & $7\ $ & $0\ $ & $1\ $ & $2\ $ & $11\ $ & $12\ $ & $13\ $ & $14\ $ & $15\ $ & $8\ $ & $9\ $ & $10 \  $ \\
	       $\ 4 \ $ & $4\ $ & $5\ $ & $6\ $ & $7\ $ & $0\ $ & $1\ $ & $2\ $ & $3\ $ & $12\ $ & $13\ $ & $14\ $ & $15\ $ & $8\ $ & $9\ $ & $10\ $ & $11 \  $ \\
	       $\ 5 \ $ & $5\ $ & $6\ $ & $7\ $ & $0\ $ & $1\ $ & $2\ $ & $3\ $ & $4\ $ & $13\ $ & $14\ $ & $15\ $ & $8\ $ & $9\ $ & $10\ $ & $11\ $ & $12 \  $ \\
	       $\ 6  \ $ & $6\ $ & $7\ $ & $0\ $ & $1\ $ & $2\ $ & $3\ $ & $4\ $ & $5\ $ & $14\ $ & $15\ $ & $8\ $ & $9\ $ & $10\ $ & $11\ $ & $12\ $ & $13 \  $ \\
	       $\ 7 \ $ & $7\ $ & $0\ $ & $1\ $ & $2\ $ & $3\ $ & $4\ $ & $5\ $ & $6\ $ & $15\ $ & $8\ $ & $9\ $ & $10\ $ & $11\ $ & $12\ $ & $13\ $ & $14 \  $ \\
	       $\ 8 \ $ & $8\ $ & $13\ $ & $10\ $ & $15\ $ & $12\ $ & $9\ $ & $14\ $ & $11\ $ & $0\ $ & $5\ $ & $2\ $ & $7\ $ & $4\ $ & $1\ $ & $6\ $ & $3 \  $ \\
	       $\ 9  \ $ & $9\ $ & $14\ $ & $11\ $ & $8\ $ & $13\ $ & $10\ $ & $15\ $ & $12\ $ & $1\ $ & $6\ $ & $3\ $ & $0\ $ & $5\ $ & $2\ $ & $7\ $ & $4 \  $ \\
	       $\ 10 \ $ & $10\ $ & $15\ $ & $12\ $ & $9\ $ & $14\ $ & $11\ $ & $8\ $ & $13\ $ & $2\ $ & $7\ $ & $4\ $ & $1\ $ & $6\ $ & $3\ $ & $0\ $ & $5 \  $ \\
	       $\ 11 \ $ & $11\ $ & $8\ $ & $13\ $ & $10\ $ & $15\ $ & $12\ $ & $9\ $ & $14\ $ & $3\ $ & $0\ $ & $5\ $ & $2\ $ & $7\ $ & $4\ $ & $1\ $ & $6 \  $ \\
	       $\ 12 \ $ & $12\ $ & $9\ $ & $14\ $ & $11\ $ & $8\ $ & $13\ $ & $10\ $ & $15\ $ & $4\ $ & $1\ $ & $6\ $ & $3\ $ & $0\ $ & $5\ $ & $2\ $ & $7 \  $ \\
	       $\ 13 \ $ & $13\ $ & $10\ $ & $15\ $ & $12\ $ & $9\ $ & $14\ $ & $11\ $ & $8\ $ & $5\ $ & $2\ $ & $7\ $ & $4\ $ & $1\ $ & $6\ $ & $3\ $ & $0 \  $ \\
	       $\ 14\ $ & $14\ $ & $11\ $ & $8\ $ & $13\ $ & $10\ $ & $15\ $ & $12\ $ & $9\ $ & $6\ $ & $3\ $ & $0\ $ & $5\ $ & $2\ $ & $7\ $ & $4\ $ & $1 \  $ \\
	       $\ 15 \ $ & $15\ $ & $12\ $ & $9\ $ & $14\ $ & $11\ $ & $8\ $ & $13\ $ & $10\ $ & $7\ $ & $4\ $ & $1\ $ & $6\ $ & $3\ $ & $0\ $ & $5\ $ & $2 \  $ \\
%\hline
\end{tabular}
\end{center}

%As we have seen in Section \ref{finagrou} a cyclic group of order $q'$ has generator $\mathbf{d}_{q'}(0,1)$ since adding modulo $q'$ this generator $q'$ times generates all elements of the group. Let us now consider a number $r$ such that $\gcd{r,q'}=1$, so that we have
%\begin{equation}
%\mathbf{d}_{q'}(0,r)=c
%\end{equation} 
%Then 

\subsection{The symmetric group of $p$ symbols $\text{Sym}_{p}$} 

We now proceed to construct the symmetric group $\text{Sym}_{p}$. This group consists of all possible $p!$ permutations of $p$ distinct symbols under composition. We have already defined a permutation $\pi_{p}$ as a bijective application $\pi_{p}:S \to S$. We first construct all permutations of the elements in $S$ and then establish the composition operation which, in this case, leads to the symmetric group. Every finite group is a subgroup of this group, which bears a crucial importance in theoretical physics.

We shall index each permutation by an integer number $0\le m \le p!-1$, such that the digits of this number in radix $p$ establish the images of the integers in $S$ under the permutation. Cauchy byline notation is usually considered for permutations. The first line indicates the element in $S$ and the second line the corresponding element in $S$ to which the corresponding one in the upper line is sent. We have
\begin{equation}
\left( \begin{array}{ccccc} 0 & 1 & \ldots & p-2 & p-1 \\ \mathbf{d}_{p}(0,\pi_{p}(m))    & \mathbf{d}_{p}(1,\pi_{p}(m)) & \ldots & \mathbf{d}_{p}(p-2,\pi_{p}(m)) & \mathbf{d}_{p}(p-1,\pi_{p}(m)) \end{array} \right) \label{permu}
\end{equation}
where it then should be clear that simply an expression like $\mathbf{d}_{p}(k,\pi_{p}(m))$ means already the element in $S$ to which the element $k$ of $S$ is sent by the m-th permutation $\pi_{p}(m)$. The latter is a non-negative integer defined as
\begin{equation}
\pi_{p}(m)=\sum_{k=0}^{p-1}p^{k}\mathbf{d}_{p}(k,\pi_{p}(m)) \label{vermut}
\end{equation}
so that, clearly, the digits of this integer in radix $p$ fully establish the form of the permutation. 

\begin{lemma} We have
\begin{eqnarray}
\pi_{p}(0)&=&p^{p}-\frac{p^{p}-p}{(p-1)^2} \label{pipe0} \\
\pi_{p}(p!-1)&=&\frac{p^{p}-p}{(p-1)^2}-1 \label{pipepe}
\end{eqnarray}
and thus $\pi_{p}(p!-1) \le \pi_{p}(m) \le \pi_{p}(0)$.
\end{lemma}

\begin{proof} Note that
\begin{eqnarray}
\pi_{p}(0)&=&\sum_{k=0}^{p-1} kp^{k} =p\frac{(p-2) p^{p}+1}{(p-1)^2}=p^{p}-\frac{p^{p}-p}{(p-1)^2} \nonumber \\
\pi_{p}(p!-1)&=&\sum_{k=0}^{p-1} k p^{p-k-1}=\frac{p^p-p^2+p-1}{(p-1)^2}=\frac{p^{p}-p}{(p-1)^2}-1  \nonumber
\end{eqnarray}
and these are the arrangements of distinct $p$ elements where the sum reaches a maximum/minimum, respectively, and hence $\pi_{p}(p!-1) \le \pi_{p}(m) \le \pi_{p}(0)$.
\end{proof}

The following result now follows easily

\begin{theor} \label{cycsym} The integers $m, n \in S$ under the operation 
\begin{equation}
\mathbf{d}_{p}\left(\mathbf{d}_{p}\left(0, m+n \right),\pi_{p}(m) \right) \label{adimod3}
\end{equation}  
form a group that is isomorphic to $C_{p}$. The isomorphism is, indeed, established by $\pi_{p}(m)$.
\end{theor}

\begin{proof} When $m=0$ we have, $\pi(0)=\sum_{k=0}^{p-1}p^{k}\mathbf{d}_{p}(k,\pi_{p}(m))=\sum_{k=0}^{p-1}p^{k}k=p^{p}-\frac{p^{p}-p}{(p-1)^{2}}$ from Eq. (\ref{pipe0}) and, then, from Eq. (\ref{adimod3})
\begin{equation}
\mathbf{d}_{p}\left(\mathbf{d}_{p}\left(0, m+n \right),\ p^{p}-\frac{p^{p}-p}{(p-1)^{2}} \right)=\mathbf{d}_{p}\left(0, m+n \right) \label{adimod2}
\end{equation} 
a result that follows directly from Eq. (\ref{Galoprop}) with $k=\mathbf{d}_{p}\left(0, m+n \right)$.
Now, if $m \ne 0$, $\mathbf{d}_{p}\left(\mathbf{d}_{p}\left(0, m+n \right),\pi_{p}(m) \right)$ gives the $\mathbf{d}_{p}\left(0, m+n \right)$-th digit of the permutation. Thus, adding $m+n$ modulo $p$ cyclically permutes the labels of the permutation and, hence, the digits. Since all of the latter are distinct, this cyclic permutation has period $p$ and, therefore, the Cayley table has identical structure as the one provided by Eq. (\ref{adimod2}) after the isomorphism induced by $\pi_{p}(m)$. 
 \end{proof}

We shall now profit from this result in finding all possible permutations of $p$ symbols which amounts to find, for each $m$, all non-negative integers $\pi_{p}(m)$.  This is achieved through the following theorem.

\begin{theor} Let $m$ be any integer such that $0 \le m \le p!-1$. The $m$-th permutation of the set of all $p!$ permutations of $p$ symbols is given by
\begin{equation}
\pi_{p}(m)=\sum_{k=0}^{p-1}p^{k}\mathbf{d}_{p}\left(\mathbf{d}_{p}\left(0, \mathbf{d}_{p}\left(0, \frac{m}{(p-1)!}\right)+k\right), (p-1)p^{p-1}+\sum_{j=0}^{p-2}p^{j}\mathbf{d}_{p-1}(j,\pi_{p-1}(m))\right) \label{permutiel}
\end{equation}
In this expression $\pi_{1}(m)=0$.
\end{theor}

\begin{proof} Let us first sketch a simple method to generate the permutations in order to give a constructive proof of the theorem. We omit the above line in Eq. (\ref{permu}) and thus a 'word' as, e.g. $2301$ means the mapping $0 \to 2$, $1 \to 3$, $2 \to 0$, $1 \to 3$. This is the so-called \emph{word representation} of a permutation \cite{Aigner2}. When there is only a single element $p=1$, we have the trivial permutation $0 \to 0$, i.e. 
\begin{equation}
\mathbf{0} \nonumber
\end{equation}   
For $p=2$ the new item $\mathbf{1}$ (we mark the new item in bold to make more visually clear the method) can now occupy two different places 
\begin{equation}
0\mathbf{1} \qquad \qquad \mathbf{1}0 \nonumber
\end{equation}   
To generate the permutations for $p=3$ we stack these permutations in a single column, we insert the new symbol  $\mathbf{2}$ at the rightmost place and permute the symbols cyclically from right to left, so that we create several columns with all permutations, i.e.
\begin{eqnarray}
&& 01\mathbf{2} \qquad \qquad 1\mathbf{2}0 \qquad \qquad \mathbf{2}01 \nonumber \\
&& 10\mathbf{2} \qquad \qquad 0\mathbf{2}1 \qquad \qquad \mathbf{2}10 \nonumber 
\end{eqnarray}   
This process can now go on indefinitely. For $p=4$, we stack again all columns of permutations in one single column and insert vertically to the right the new element $\mathbf{3}$, subsequently permuting the resulting arrangement cyclically
\begin{eqnarray}
&& 012\mathbf{3} \qquad \qquad 12\mathbf{3}0 \qquad \qquad 2\mathbf{3}01 \qquad \qquad \mathbf{3}012  \nonumber \\
&& 102\mathbf{3} \qquad \qquad 02\mathbf{3}1 \qquad \qquad 2\mathbf{3}10 \qquad \qquad \mathbf{3}102  \nonumber \\
&& 120\mathbf{3} \qquad \qquad 20\mathbf{3}1 \qquad \qquad 0\mathbf{3}12 \qquad \qquad \mathbf{3}120  \nonumber \\
&& 021\mathbf{3} \qquad \qquad 21\mathbf{3}0 \qquad \qquad 1\mathbf{3}02 \qquad \qquad \mathbf{3}021  \nonumber \\
&& 201\mathbf{3} \qquad \qquad 01\mathbf{3}2 \qquad \qquad 1\mathbf{3}20 \qquad \qquad \mathbf{3}201  \nonumber \\
&& 210\mathbf{3} \qquad \qquad 10\mathbf{3}2 \qquad \qquad 0\mathbf{3}21 \qquad \qquad \mathbf{3}210  \nonumber
\end{eqnarray}   

By using the factorial number system mentioned in the remarks following the proof of Theorem $\ref{numeralsys}$, we can define a number $m$ in this system as
\begin{equation}
m=\sum_{k=1}^{p}\mathbf{d}_{k}\left(0, \frac{m}{(k-1)!}\right)(k-1)!
\end{equation}
This number thus may take any value between $0$ and $p!-1$. We can now establish a bijection between all these possible values for $m$ and all permutations formed above at stage $p$. We define the $k$-th digit of $m$, $\mathbf{d}_{k}\left(0, \frac{m}{(k-1)!}\right)$  through the number of the column in the table at stage $k$ (0 being the leftmost column and $k-1$ the rightmost one) in which the symbol $k-1$ (marked in bold at that stage) is found. 

\noindent (\emph{Example:} To find $m$ for the permutation 2301 above $(p=4)$, we note that this permutation is gradually constructed as: $0 \to 01 \to 201 \to 2301$, for which the corresponding $k$'s are 0, 0, 2, 2. Thus, we have $m=0\cdot 4^{0}+0\cdot 4^{1}+2\cdot 4^{2}+2\cdot 4^{3}=160$).

The permutations are each recursively constructed starting by $0$ and then inserting the symbols $1$, $2$, $\ldots$ to the right of the previously formed permutations, and then permuting the resulting column cyclically. Then $\pi_{p-1}(m)$, i.e. the permutation formed at stage $p-1$, is given by
\begin{equation}
\pi_{p-1}(m)=\sum_{j=0}^{p-2}(p-1)^{j}\mathbf{d}_{p-1}(j,\pi_{p-1}(m))
\end{equation}
At stage $p$ the leftmost permutation is formed by extracting the digits of  $\pi_{p-1}(m)$ and adjoining the symbol $p-1$ to the right of the chain, then we have
\begin{equation}
(p-1)p^{p-1}+\sum_{j=0}^{p-2}p^{j}\mathbf{d}_{p-1}(j,\pi_{p-1}(m))
\end{equation}
for the rightmost permutation. The permutation $\pi_{p}(m)$ is thus obtained by cyclically permuting this rightmost permutation a number of times $\mathbf{d}_{p}\left(0, \frac{m}{(p-1)!}\right)$. Thus, compared to the rightmost permutation for which $\mathbf{d}_{p}\left(0, \frac{m}{(p-1)!}\right)=0$ the k-th digit of the permutation corresponds to the place $\mathbf{d}_{p}\left(0, \mathbf{d}_{p}\left(0, \frac{m}{(p-1)!}\right)+k\right)$ and, therefore, the expression $\mathbf{d}_{p}\left(\mathbf{d}_{p}\left(0, \mathbf{d}_{p}\left(0, \frac{m}{(p-1)!}\right)+k\right), (p-1)p^{p-1}+\sum_{j=0}^{p-2}p^{j}\mathbf{d}_{p-1}(j,\pi_{p-1}(m))\right)$ outputs the value to which $k$ is sent by the permutation $\pi_{p}(m)$ and, hence, is equal to $\mathbf{d}_{p}(k,\pi_{p}(m))$. Thus, by Eq. (\ref{vermut}) we obtain the result of the theorem. \end{proof}

\begin{rem} Although the expression Eq. (\ref{permutiel}) above in terms of the digit function is original, the construction of the permutations is closely related to the classical ones by Burnside \cite{Burnside} and Moore \cite{Moore2} (see also \cite{Coxeter2}) and one certainly sees the action of the two generators $(0\ 1 \ \ldots p-1)$ and $(0 1)$ (employed by those authors) during the construction process. We are, however, not presenting the groups through algebraic relationships satisfied by the generators as in \cite{Coxeter2} but in a different, fully explicit manner, where no particular reference to generators is made and where one does know how to operate with arbitrary elements within the group, having visible the whole abstract structure of the group at the same time. 
\end{rem}

We thus know $\pi_{p}(m)$ for any given $m \in [0,p!-1]$ integer. Since the mapping between $m$ and all possible bifurcations of $p$ symbols is bijective, we can ask whether given $\pi_{p}(m)$, $m$ can be known, i.e. the inverse of the mapping. The answer, of course, is affirmative and is given by the following lemma.

\begin{lemma} The following expression holds
\begin{equation}
m=\sum_{n=0}^{p!-1}n\delta_{\pi_{p}(n)\pi_{p}(m)}
 \label{invem}
\end{equation}
where both $\pi_{p}(n)$ and $\pi_{p}(m)$ are any permutation of $p$ symbols given by Eq. (\ref{permutiel}).
\end{lemma} 
%m=\sum_{n=0}^{p!-1}n\mathbf{d}_{p!}\left(\pi_{p}(n), (p!)^{\pi_{p}(m)}\right)

\begin{proof} The Kronecker delta in the sum of Eq. (\ref{invem}) can only output one if $\pi_{p}(m)$ and $\pi_{p}(n)$ are equal. Now, because of the uniqueness of the radix-$p$ representation for nonnegative integers \cite{Andrews} all digits of $\pi_{p}(m)$ and $\pi_{p}(n)$ must necessarily be equal as well, but then this means that $\pi_{p}(m)$ and $\pi_{p}(n)$ are the same permutation and, hence, that $m=n$. Thus, the sum returns only that $n$ which is equal to $m$.
\end{proof}

%\begin{proof} By using Eq. (\ref{krone}) we have
%\begin{equation}
%\sum_{n=0}^{p!-1}n\mathbf{d}_{p!}\left(\pi_{p}(n), (p!)^{\pi_{p}(m)}\right)=
%\sum_{n=0}^{p!-1}n\delta_{\pi_{p}(n),\pi_{p}(m)}
%\end{equation}
%Now, since the Kronecker delta outputs 1 only if $\pi_{p}(n)=\pi_{p}(m)$ and zero otherwise, by taking into account that all permutations are distinct, it follows that, necessarily, only the term $n=m$ contributes to the sum.
%\end{proof}

\begin{theor} The permutation $\pi_{p}(m)$ has inverse under composition, which is also a permutation $\pi_{p}^{-1}(m)$, given by
\begin{equation}
\pi_{p}^{-1}(m)=\sum_{k=0}^{p-1}p^{k}\sum_{n=0}^{p-1}n\mathbf{d}_{p}(\mathbf{d}_{p}(n,\pi_{p}(m)), p^{k}) \label{invep}
\end{equation}
so that $\pi_{p}(m) \circ \pi_{p}^{-1}(m) = \pi_{p}^{-1}(m) \circ \pi_{p}(m)=\pi_{p}(0)$, where $\pi_{p}(0)=\sum_{k=0}^{p-1}p^{k}k$ is the identity permutation (which fixes all elements).
\end{theor}

\begin{proof} We have, 
\begin{equation}
\pi_{p}^{-1}(m) \circ \pi_{p}(m)=\sum_{k=0}^{p-1}p^{k}\sum_{n=0}^{p-1}n\mathbf{d}_{p}(\mathbf{d}_{p}(n,\pi_{p}(m)), p^{\mathbf{d}_{p}(k,\pi_{p}(m))})
\end{equation}
since $k \to \mathbf{d}_{p}(k,\pi_{p}(m))$ under the action of $\pi_{p}(m)$. Now, by using Eq. (\ref{krone}), we obtain
\begin{equation}
\pi_{p}^{-1}(m) \circ \pi_{p}(m)=\sum_{k=0}^{p-1}p^{k}\sum_{n=0}^{p-1}n
\delta_{\mathbf{d}_{p}(n,\pi_{p}(m)),\mathbf{d}_{p}(k,\pi_{p}(m)}
\end{equation}
but, since $\pi_{p}(m)$ is bijective, it is also injective, preserving distinctness of its elements under the mapping. Now, for the Kronecker delta to output a nonzero value $\mathbf{d}_{p}(n,\pi_{p}(m))=\mathbf{d}_{p}(k,\pi_{p}(m))$. But this is only possible if $n=k$ because all digits of $\pi_{p}(m)$ are distinct. Hence, we obtain
\begin{equation}
\pi_{p}^{-1}(m) \circ \pi_{p}(m)=\sum_{k=0}^{p-1}p^{k}\sum_{n=0}^{p-1}n\delta_{nk}=\sum_{k=0}^{p-1}p^{k}k=\pi_{p}(0)
\end{equation}
which proves the result (it is trivial to prove that $\pi_{p}(m) \circ \pi_{p}^{-1}(m)=\pi_{p}(0)$ as well). \end{proof}

\begin{theor} The symmetric group of $p$ symbols $\text{\emph{Sym}}_{p}$, is formed by all $p!$ permutations $\pi_{p}(m)$ obtained from Eq. (\ref{permutiel}) under the operation
\begin{equation}
\sum_{k=0}^{p-1}p^{k}\mathbf{d}_{p}\left(\mathbf{d}_{p}\left(k, \pi_{p}(n)\right), \pi_{p}(m) \right) 
\label{simetron}
\end{equation}
\end{theor}

\begin{proof} The digit $k$ is sent by the permutation $\pi_{p}(n)$ to $k'=\mathbf{d}_{p}\left(k, \pi_{p}(n)\right)$, which is again an element of $S$. The latter is sent by the permutation $\pi_{p}(m)$ to $k''=\mathbf{d}_{p}\left(k', \pi_{p}(m)\right)=\mathbf{d}_{p}\left(\mathbf{d}_{p}\left(k, \pi_{p}(n)\right), \pi_{p}(m)\right)$. Thus, the permutation resulting from the composition of the two permutations $ \pi_{p}(m)\circ \pi_{p}(n)$ is given by Eq. (\ref{simetron}). The same result can be easily obtained from the composition-decomposition theorem (Theorem \ref{compy}). The set of all permutations under composition is the symmetric group $\text{Sym}_{p}$. 
\end{proof}

We now construct with Eqs. (\ref{permutiel}) and (\ref{simetron}) the symmetric groups of small order. We observe from the above that there are two equivalent ways of representing permutations: 1) one can give the natural number $\pi_{p}(m)$; 2) one can give a word with the digits $\mathbf{d}_{p}(k,\pi_{p}(m))$ concatenated. Both representations are trivially related through Eq. (\ref{vermut}). In constructing the Cayley table for the groups we can also choose 3) to simply indicate $\pi_{p}(0)$, $\pi_{p}(1)$,... $\pi_{p}(p!-1)$  in the table, with the understanding that Eq. (\ref{permutiel}) provides such permutation. These three equivalent representations are illustrated in the following construction of  $\text{Sym}_{2}$ of $2$ elements (which is trivially isomorphic to $C_{2}$).

\begin{center}
\begin{tabular}{c|cc}
%\hline 
$\ \ \text{Sym}_{2}  \ \ $ & $ \ 2  \ $ & $\ \  1 \ \ $ \\
\hline
$\ 2 \ $ & $\ 2 \ $ & $\  1 \ $ \\
$\ 1 \ $ & $\ 1 \ $ & $ \ 2 \ $ \\
%\hline
\end{tabular}
\qquad 
\begin{tabular}{c|cc}
%\hline 
$\ \ \text{Sym}_{2}  \ \ $ & $ \ 01  \ $ & $\ \  10 \ \ $ \\
\hline
$\ 01 \ $ & $\ 01 \ $ & $\  10 \ $ \\
$\ 10 \ $ & $\ 10 \ $ & $ \ 01 \ $ \\
%\hline
\end{tabular}
\qquad 
\begin{tabular}{c|cc}
%\hline 
$\ \ \text{Sym}_{2}  \ \ $ & $ \ \pi_{2}(0)  \ $ & $\ \  \pi_{2}(1) \ \ $ \\
\hline
$\ \pi_{2}(0) \ $ & $\ \pi_{2}(0) \ $ & $\  \pi_{2}(1) \ $ \\
$\ \pi_{2}(1) \ $ & $\ \pi_{2}(1) \ $ & $ \ \pi_{2}(0) \ $ \\
%\hline
\end{tabular}
\end{center}

The leftmost of these tables is directly the output of Eq. (\ref{simetron}) for input values $m=0, 1$ (from top to bottom in the rows) and $n=0,1$ (from left to right in the columns), since one has, from Eq. (\ref{permutiel}) 
\begin{eqnarray}
\pi_{2}(0)&=&0\cdot 2^{0}+1\cdot 2^{1}=2 \nonumber \\
\pi_{2}(1)&=&1\cdot 2^{0}+0\cdot 2^{1}=1 \nonumber
\end{eqnarray}
This clarifies the three tables above. The central one explicitly provides the arrangement of the digits and thus conveys more information.

The symmetric group $\text{Sym}_{3}$ is also easily constructed from Eqs. (\ref{permutiel}) and (\ref{simetron}) as

\begin{center}
\begin{tabular}{c|cccccc}
%\hline 
$\  \text{Sym}_{3}  \ $ & $ 21 \ $ & $\ 19 \ $ & $\   7 \  $ & $ 15 \ $ & $\ 11 \ $ & $\   5 \  $ \\
\hline
               $\ 21  \ $ & $\ 21 \ $ & $\ 19 \ $ & $\   7 \  $ & $ 15 \ $ & $\ 11 \ $ & $\   5 \  $ \\
               $\ 19  \ $ & $\ 19 \ $ & $\ 21 \ $ & $\   15 \  $ & $ 7 \ $ & $\ 5 \ $ & $\   11 \  $ \\
	       $\ 7  \ $ & $\ 7 \ $ & $\ 5 \ $ & $\   11 \  $ & $ 19 \ $ & $\ 21 \ $ & $\  15 \  $ \\
	       $\ 15  \ $ & $\ 15 \ $ & $\ 11 \ $ & $\  5 \  $ & $ 21 \ $ & $\ 19 \ $ & $\   7 \  $ \\
	       $\ 11  \ $ & $\ 11 \ $ & $\ 15 \ $ & $\   21 \  $ & $ 5 \ $ & $\ 7 \ $ & $\  19 \  $ \\
	       $\ 5  \ $ & $\ 5 \ $ & $\ 7 \ $ & $\  19 \  $ & $ 11 \ $ & $\ 15 \ $ & $\  21 \  $ \\
%\hline
\end{tabular}
\qquad
\begin{tabular}{c|cccccc}
%\hline 
$\  \text{Sym}_{3}  \ $ & $ \ 012 \ $ & $\ 102 \ $ & $\   120 \  $ & $ 021 \ $ & $\ 201 \ $ & $\  210 \  $ \\
\hline    
		$\ 012  \ $ &  $\ 012 \ $  & $\ 102 \ $ & $\   120 \  $ & $ 021 \ $ & $\ 201 \ $ & $\  210 \  $ \\
               $\ 102  \ $ & $\ 102 \ $  & $\ 012 \ $ & $\  210 \  $ & $ 201 \ $ & $\ 021 \ $ & $\   120 \  $ \\
               $\ 120  \ $ & $\ 120 \ $  & $\ 021 \ $ & $\   201 \  $ & $ 210 \ $ & $\ 012 \ $ & $\   102 \  $ \\
	       $\ 021  \ $ & $\ 021 \ $  & $\ 120 \ $ & $\   102 \  $ & $ 012 \ $ & $\ 210 \ $ & $\  201 \  $ \\
	       $\ 201  \ $ & $\ 201 \ $  & $\ 210 \ $ & $\  012 \  $ & $ 102 \ $ & $\ 120 \ $ & $\   021 \  $ \\
	       $\ 210  \ $ & $\ 210 \ $  & $\ 201 \ $ & $\   021 \  $ & $ 120 \ $ & $\ 102 \ $ & $\  012 \  $ \\
%\hline
\end{tabular}
\end{center}

\begin{figure*} %Shown is a window where $t \in [0,150]$ and $i\in [1, 150]$
\includegraphics[width=0.3 \textwidth]{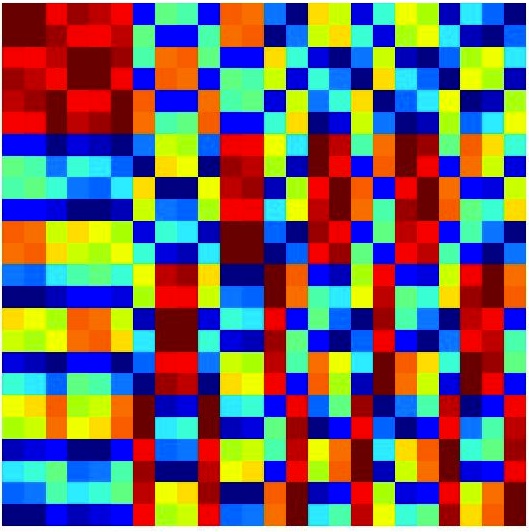}
\caption{\scriptsize{The Cayley table of the symmetric group $\text{Sym}_{4}$ of $p=4$ symbols and order $24$, obtained from Eq. (\ref{simetron}).~\\ ~~~~~~~~~~~~~~~~~~~\\ ~~~~~~~~~~~~~~~~~~~~~~~~~}} \label{sym4}
%\end{figure*}
%\begin{figure*} %Shown is a window where $t \in [0,150]$ and $i\in [1, 150]$

\includegraphics[width=0.5 \textwidth]{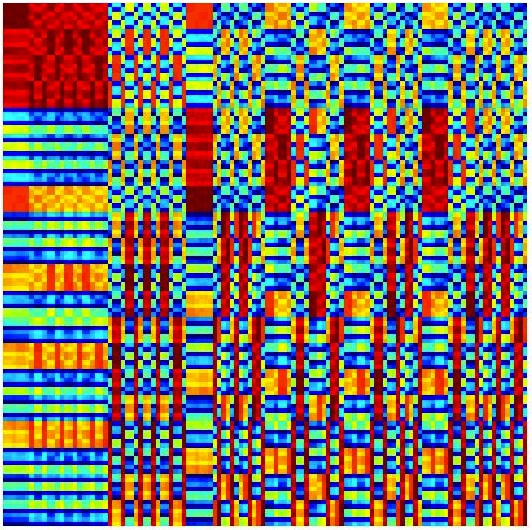}
\caption{\scriptsize{The Cayley table of the symmetric group $\text{Sym}_{5}$ of $p=5$ symbols and order $120$,  obtained from Eq. (\ref{simetron}).}} \label{sym}
\end{figure*}

Although $\pi_{p}(m)$ does not change monotonically with $m$, $m$ always increases from $0$ to $p-1$ in the first row, from left to right. There is thus still another equivalent way of representing the symmetric group, by giving in the tables just the value $m$ that maps each entry to $\pi(m)$. This table can be obtained by applying Eq. (\ref{invem}) to each entry of the left table above these lines. We thus, equivalently, obtain the isomorphic table
\begin{center}
\begin{tabular}{c|cccccc}
%\hline 
$\  \text{Sym}_{3}  \ $ & $ 0 \ $ & $\ 1 \ $ & $\  2 \  $ & $\ 3 \ $ & $\ 4 \ $ & $\   5 \  $ \\
\hline
               $\ 0  \ $ & $\ 0 \ $ & $\ 1 \ $ & $\  2 \  $ & $\ 3 \ $ & $\ 4 \ $ & $\   5 \  $ \\
               $\ 1  \ $ & $\ 1 \ $ & $\ 0 \ $ & $\   3 \  $ & $\ 2 \ $ & $\ 5 \ $ & $\   4 \  $ \\
	       $\ 2  \ $ & $\ 2 \ $ & $\ 5 \ $ & $\   4 \  $ & $\ 1 \ $ & $\ 0 \ $ & $\  3 \  $ \\
	       $\ 3  \ $ & $\ 3 \ $ & $\ 4 \ $ & $\  5 \  $ & $\ 0 \ $ & $\ 1 \ $ & $\   2 \  $ \\
	       $\ 4  \ $ & $\ 4 \ $ & $\ 3 \ $ & $\   0 \  $ & $\ 5 \ $ & $\ 2 \ $ & $\  1 \  $ \\
	       $\ 5  \ $ & $\ 5 \ $ & $\ 2 \ $ & $\  1 \  $ & $\ 4 \ $ & $\ 3 \ $ & $\  0 \  $ \\
%\hline
\end{tabular}
\end{center}

The symmetric group $\text{Sym}_{3}$ is isomorphic to the dihedral group of degree $q=3$ and order $p=6$, $D_{6}$,  derived above. For the symmetric groups $\text{Sym}_{4}$ (the group of rotations that fix a cube or a octahedron) \cite{Ma} and $\text{Sym}_{5}$ we present the Cayley tables as color codes.  The red colored places represent those permutations that fix the element $p-1$ and the bluish places correspond to permutations that send the element $p-1$ to $0$. In this way the darkest red permutation corresponds to the identity $\pi_{p}(0)$ and the darkest blue permutation to $\pi_{p}(p!-1)$.

\subsection{The alternating group of $p$ symbols $A_{p}$}

A permutation has even/odd signature if it can be reached by an even/odd number of transpositions. We follow Passman \cite{Passman} in introducing the following definition, adapted to our notation.

\begin{theor}\label{signa} Let the signature $\sigma_{\pi_{p}(m)}$ of a permutation $\pi_{p}(m)$ given by Eq. (\ref{permutiel}) be $1$ if it is even and $0$ if it is odd. Then,  we have 
\begin{equation}
\sigma_{\pi_{p}(m)}=\prod_{j=1}^{p-1}\prod_{k=0}^{j-1}\frac{j-k}{\mathbf{d}_{p}(j, \pi_{p}(m))-\mathbf{d}_{p}(k, \pi_{p}(m))} \label{signa}
\end{equation}
\end{theor}

\begin{proof} First note that the products over $j$ and $k$ scan all $p(p-1)/2={p \choose 2}$ possible choices of 2 distinct elements $j$ and $k$ out of the set of integers in the interval $[0,p-1]$ (the order being irrelevant). Thus, the denominator can never be zero and, therefore, $\sigma_{\pi_{p}(m)}$ is well defined.

If we consider the identity permutation $\pi_{p}(0)$, we clearly have
\begin{equation}
\sigma_{\pi_{p}(0)}=\prod_{j=1}^{p-1}\prod_{k=0}^{j-1}\frac{j-k}{j-k}=1 \label{signaide}
\end{equation}

Now, let us consider the permutation for $m'$, which introduces a transposition on the $j_{0}$ and $k_{0}$ elements on the identity. We have
\begin{equation}
\sigma_{\pi_{p}(m')}=\frac{j_{0}-k_{0}}{k_{0}-j_{0}}\prod_{j=1, j \ne j_{0}}^{p-1}\prod_{k=0, k\ne k_{0}}^{j-1}\frac{j-k}{j-k}=-1=-\sigma_{\pi_{p}(0)} \label{signatra}
\end{equation}

All permutations can be reached from the identity through transpositions, as we have already seen in constructing the permutations through cycles. Since each permutation is located in a specific place in the table of permutations, where all them can be reached through transpositions, with the identity being at location $0$, it is clear that Eq. (\ref{signa}) gives the signature of the permutation. The latter is, by construction, independent of the path in permutation space taken to reach it from the identity through composition of transpositions.  \end{proof}

This leads to the straightforward construction of the group formed by all even permutations, called the alternating group $A_{p}$ of order $p!/2$. That this is a group follows simply by the fact that the subgroup of all even permutations is closed and contains the identity $\pi_{p}(0)$. The inverse of each element is also there since the inverse of an even permutation can be trivially shown to be an even permutation as well. Finally, composition of permutations is always associative. Note that the group is generally non-commutative. For all $p \ge 5$, alternating groups are also known to be \emph{simple}: they have no normal subgroups. We now construct the infinite family of alternating groups.

\begin{theor} The alternating group of $p$ symbols $A_{p}$, is formed by all $p!/2$ even permutations $\pi_{p}(m)$ obtained from Eq. (\ref{permutiel}) under the operation Eq. (\ref{simetron}) and the further constraint
\begin{equation}
\sigma_{\pi_{p}(m)}=1
%\frac{1+\sigma_{\pi_{p}(m)}}{2}=1
\label{simetronal}
\end{equation}
for all $m$ integer $\in [0, p!/2-1]$.
\end{theor}

The first interesting alternating group is $A_{4}$. It contains the rotational symmetries of the tetrahedron. $A_{4}$ plays an important role as model in flavor physics for understanding quark and neutrino mixing angles \cite{He}. The Cayley table of the group, as provided by Eqs. (\ref{permutiel}), (\ref{simetron}) and (\ref{simetronal}) is given in terms of the word representation of the permutations as 
~\\

\begin{center}
\begin{tabular}{c|cccccccccccc}
%\hline 
$\  A_{4}  \ $ & $\ 0123\ $ & $ \ 1203\ $ & $ \ 2013\ $ & $ \ 0231\ $ & $ \ 2130\ $ & $ \ 1032\ $ & $ \ 2301\ $ & $ \ 0312\ $ & $ \ 1320\ $ & $ \ 3102\ $ & $ \ 3021\ $ & $ \ 3210 \  $ 
\\
\hline
$\  0123  \ $ & $\ 0123\ $ & $ \ 1203\ $ & $ \ 2013\ $ & $ \ 0231\ $ & $ \ 2130\ $ & $ \ 1032\ $ & $ \ 2301\ $ & $ \ 0312\ $ & $ \ 1320\ $ & $ \ 3102\ $ & $ \ 3021\ $ & $ \ 3210 \  $ 
\\
$\  1203  \ $ & $\ 1203\ $ & $ \ 2013\ $ & $ \ 0123\ $ & $ \ 1032\ $ & $ \ 0231\ $ & $ \ 2130\ $ & $ \ 0312\ $ & $ \ 1320\ $ & $ \ 2301\ $ & $ \ 3210\ $ & $ \ 3102\ $ & $ \ 3021 \  $ 
\\
$\  2013  \ $ & $\ 2013\ $ & $ \ 0123\ $ & $ \ 1203\ $ & $ \ 2130\ $ & $ \ 1032\ $ & $ \ 0231\ $ & $ \ 1320\ $ & $ \ 2301\ $ & $ \ 0312\ $ & $ \ 3021\ $ & $ \ 3210\ $ & $ \ 3102 \  $ 
\\
$\  0231  \ $ & $\ 0231\ $ & $ \ 2301\ $ & $ \ 3021\ $ & $ \ 0312\ $ & $ \ 3210\ $ & $ \ 2013\ $ & $ \ 3102\ $ & $ \ 0123\ $ & $ \ 2130\ $ & $ \ 1203\ $ & $ \ 1032\ $ & $ \ 1320 \  $ 
\\
$\  2130  \ $ & $ \ 2130\ $ & $ \ 1320\ $ & $ \ 3210\ $ & $ \ 2301\ $ & $ \ 3102\ $ & $ \ 1203\ $ & $ \ 3021\ $ & $ \ 2013\ $ & $ \ 1032\ $ & $ \ 0123\ $ & $ \ 0231\ $ & $ \ 0312 \  $ 
\\
$\ 1032  \ $ & $\ 1032\ $ & $ \ 0312\ $ & $ \ 3102\ $ & $ \ 1320\ $ & $ \ 3021\ $ & $ \ 0123\ $ & $ \ 3210\ $ & $ \ 1203\ $ & $ \ 0231\ $ & $ \ 2013\ $ & $ \ 2130\ $ & $ \ 2301 \  $ 
\\
$\ 2301  \ $ & $\ 2301\ $ & $ \ 3021\ $ & $ \ 0231\ $ & $ \ 2013\ $ & $ \ 0312\ $ & $ \ 3210\ $ & $ \ 0123\ $ & $ \ 2130\ $ & $ \ 3102\ $ & $ \ 1320\ $ & $ \ 1203\ $ & $ \ 1032 \  $ 
\\
$\ 0312  \ $ & $\ 0312\ $ & $ \ 3102\ $ & $ \ 1032\ $ & $ \ 0123\ $ & $ \ 1320\ $ & $ \ 3021\ $ & $ \ 1203\ $ & $ \ 0231\ $ & $ \ 3210\ $ & $ \ 2301\ $ & $ \ 2013\ $ & $ \ 2130 \  $ 
\\
$\ 1320  \ $ & $\ 1320\ $ & $ \ 3210\ $ & $ \ 2130\ $ & $ \ 1203\ $ & $ \ 2301\ $ & $ \ 3102\ $ & $ \ 2013\ $ & $ \ 1032\ $ & $ \ 3021\ $ & $ \ 0312\ $ & $ \ 0123\ $ & $ \ 0231 \  $ 
\\
$\  3102  \ $ & $\ 3102\ $ & $ \ 1032\ $ & $ \ 0312\ $ & $ \ 3021\ $ & $ \ 0123\ $ & $ \ 1320\ $ & $ \ 0231\ $ & $ \ 3210\ $ & $ \ 1203\ $ & $ \ 2130\ $ & $ \ 2301\ $ & $ \ 2013 \  $ 
\\
$\  3021  \ $ & $\ 3021\ $ & $ \ 0231\ $ & $ \ 2301\ $ & $ \ 3210\ $ & $ \ 2013\ $ & $ \ 0312\ $ & $ \ 2130\ $ & $ \ 3102\ $ & $ \ 0123\ $ & $ \ 1032\ $ & $ \ 1320\ $ & $ \ 1203 \  $ 
\\
$\  3210  \ $ & $\ 3210\ $ & $ \ 2130\ $ & $ \ 1320\ $ & $ \ 3102\ $ & $ \ 1203\ $ & $ \ 2301\ $ & $ \ 1032\ $ & $ \ 3021\ $ & $ \ 2013\ $ & $ \ 0231\ $ & $ \ 0312\ $ & $ \ 0123 \  $ 
\\
%\hline
\end{tabular}
\end{center}

%\begin{center}
\begin{figure*} %Shown is a window where $t \in [0,150]$ and $i\in [1, 150]$
\includegraphics[width=0.3 \textwidth]{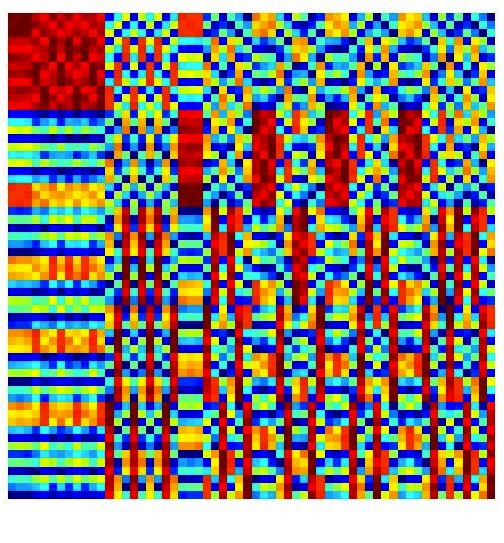}
\includegraphics[width=1.0 \textwidth]{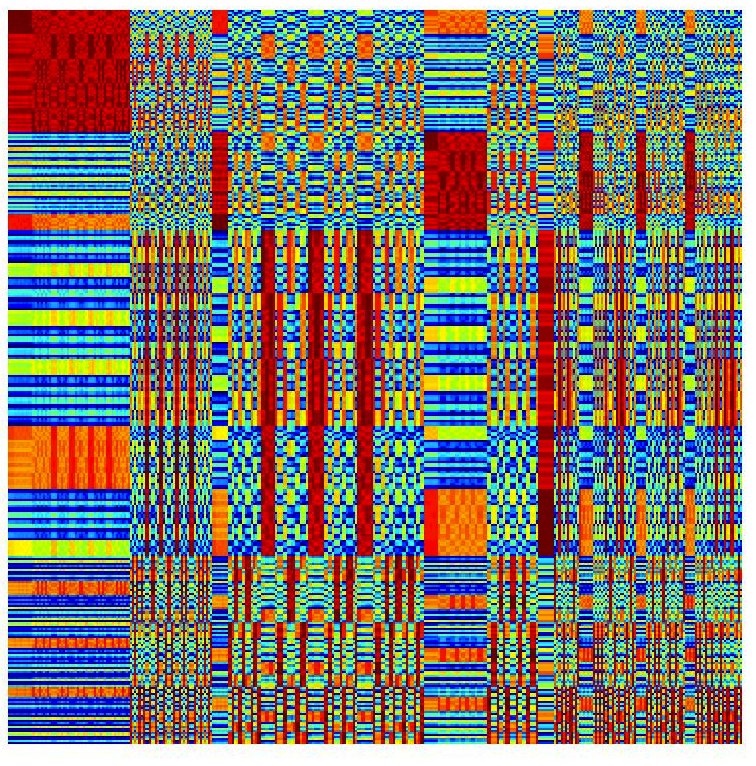}
\caption{\scriptsize{The Cayley table of the alternating group $A_{5}$ of $p=5$ symbols and order $60$ (top) and $A_{6}$ of $p=6$ symbols and order $360$, obtained from Eqs. (\ref{permutiel}), (\ref{simetron}) and (\ref{simetronal}).}} \label{altern5}
\end{figure*}

%\end{center}

The Cayley tables of $A_{5}$ and $A_{6}$ (both are simple groups \cite{ConwayATLAS}) are displayed in Fig. \ref{altern5}. The alternating group $A_{5}$ is the group of rotations of an icosahedron or a dodecahedron \cite{Ma}. Together with $A_{4}$ and $\text{Sym}_{4}$ derived above, $A_{5}$ constitute the polyhedral groups, which are the groups of rotations that map the Platonic solids onto themselves \cite{Hamermesh, Ma}. The polyhedral groups and the dihedral groups $D_{2q}$, whose Cayley tables have also been obtained in this article, are all subgroups of the special unitary group SU(3) \cite{Fairbairn, Ludl, Hamermesh}.

\section{Conclusion}

In this article, explicit expressions have been found for the Cayley tables of infinite families of finite groups including cyclic groups and their direct sums (i.e. \emph{all} abelian finite groups), dihedral, dicyclic and the whole family of metacyclic groups where they are contained and, finally, the symmetric and alternating groups. The interest of having explicit expressions for the Cayley tables is motivated by recent results \cite{arxiv3} where fractal discontinuous curves and surfaces have been constructed so that their ordinary addition is equal everywhere to a prescribed function (that can be continuous and differentiable). All these mathematical methods, although still in a rather abstract stage of development, are hoped to find applications in theoretical particle physics, statistical mechanics and chemical physics.

We note that the Cayley table contains all information of a given group. From the Cayley table it is immediate to determine certain subgroups, as e.g. the center of the group, and regular representations can automatically be constructed. The mathematical expressions also provide new routes to find the conjugacy classes, irreducible representations and character tables, all them derivable from the Cayley tables and from the so-called 'great orthogonality theorem' as well as the 'celebrated theorem' of group theory \cite{Tinkham}. 

Our approach is an alternative to the classical ones (which focus on the group generators and the algebraic equations that they satisfy \cite{Coxeter2}). It makes use of a digit function, whose mathematical properties have been also further investigated here, following our previous works \cite{arxiv2, arxiv3}. The digit function is a central concept in our recent formulation of quantum mechanics \cite{QUANTUM} and we have shown how naturally does it relate to group theory.

\bibliography{biblos}{}
\bibliographystyle{h-physrev3.bst}

\end{document}